\documentclass[12pt,letterpaper]{article}

\usepackage[T1]{fontenc}
\usepackage{microtype}
\usepackage{graphicx,natbib,enumerate,amsthm,subfig,hyperref,easyReview}
\setcitestyle{citesep={,}}
\hypersetup{colorlinks=true,
	linkcolor=blue,
	citecolor=blue,
	pdfauthor = Ted To,
	pdftitle = {Renewable Natural Resources, Regime Shift and Hysteresis},
	pdfsubject = Working paper,
	pdfkeywords = {working paper, draft},
	pdfproducer = pdfLaTeX 3.141592653,
	pdfcreator = TeXstudio}
\usepackage[normalem]{ulem}
\usepackage{dutchcal}
\usepackage{mystyle}
\usepackage{cleveref}
\usepackage{caption}
\usepackage{appendix}

\usepackage{lmodern}

\newtheorem*{remark*}{Remark}

\author{Ted To%
}

\title{\bf\LARGE Renewable Natural Resources, Regime Shift and Hysteresis\thanks{
		I thank Len Mirman, Bruno Nkuiya, Erick Sager, Partha Sen and two anonymous referees for helpful comments and suggestions.  %
	}
}
\date{October 2026}

\begin{document}
	\maketitle
	\begin{abstract}%
		Many of the world's renewable resources are in decline. Optimal harvests with smooth recruitment is well studied but in recent years, ecologists have concluded that tipping points in recruitment are common. Recruitment with a tipping point has low-fecundity below the tipping point and high-fecundity above. When the discounted value of high-fecundity is sufficiently high, there is a high-fecundity steady-state. This steady-state is stable but in some cases, small perturbations may result in large, temporary reductions in recruitment and harvests.
		Below the tipping point, a low-fecundity steady-state need not exist. When a low-fecundity steady-state does exist, there is an endogenous tipping (Skiba) point: below, harvests converge to the low-fecundity steady-state and above, an austere harvest policy transitions the renewable resource to high-fecundity recruitment.  If there is hysteresis in recruitment, the high steady-state may not be stable. Moreover, if the high-/low-fecundity differential is large then following a downward perturbation, fecundity optimally remains low.\bigskip
		
		\parindent0pt Key words: renewable resource management, tipping point, hysteresis, regime shift 
	\end{abstract}
	
	\thispagestyle{empty}
	\setcounter{page}{0}
	\newpage
	
	\section{Introduction}

	Many of the world's renewable resources are in decline, including fisheries \citep{worm:impacts,jackson:ecological}, forests \citep{fao:state} and wildlife \citep{felbab-brown:extinction}. The theoretical literature on modeling renewable resources has a long history, dating back to \cite{gordon:economic, scott:fishery, smith:economics} and \cite{clark:economics,clark:profit}. While the focus of these works has often been on fisheries, the modeling is applicable to renewable resources in general \citep{clark:mathematical}.\footnote{Indeed, these models are ecological variants of the classical optimal growth model \citep{ramsey:mathematical, cass:optimum, koopmans:concept}.}
	
	The ``smooth recruitment function'' renewable resource problem is well studied \citep{clark:mathematical}, however, it is now believed that many renewable resources are subject to ``tipping points'' in recruitment. For marine resources, there is a consensus that tipping points are important \citep{selkoe:principles, hunsicker:characterizing}. For instance, minimal genetic diversity is required for effective reproduction \citep{kardos:crucial}. %
	For tropical rain forests, the tipping mechanism results from changes in rainfall patterns due to deforestation that transitions the ecosystem from rain forest to savanna \citep{nobre:tipping, malhado:cerrado}.

	Much of the prior work on renewable resources with tipping points models the tipping process using a hazard model.  In a hazard model, there is no fixed tipping point and the likelihood of tipping is non-increasing in the resource stock. This has the interpretation that, at lower stock levels, the ecosystem may be less resilient and adverse shocks may be more likely \citep{reed:optimal, polasky:optimal, zeeuw:managing}. Uncertainty in the hazard model is exogenous to the model in the sense that even if the renewable resource stock remains stationary, a positive hazard rate implies that the renewable resource will eventually tip. 
	
	In this paper, I characterize the optimal harvest of a renewable resource in the presence of tipping points. %
	A recruitment function governs the natural growth rate of the resource stock. 
	Instead of modeling the tipping process using a hazard model, I assume that the location of the tipping point is fixed.\footnote{
		This is not without precedent.  For instance, in their diagrams illustrating regime shift, \cite{dudgeon:phase} and \cite{selkoe:principles} show fixed tipping points. %
	} I consider my approach to be complementary with the hazard literature. Fixed tipping points allow me to consider two important aspects of tipping points. The first is to allow for the possibility that with sufficiently austere harvests, a low fecundity renewable resource is able to recover. Second, given that the renewable resource is able to recover, it then becomes possible to tractably model hysteresis.

	In my model, when there is no hysteresis and when the incremental value of high-fecundity is sufficiently large, there exists a high-fecundity steady-state.
	This high steady-state is stable in the sense that following a small perturbation, the resource stock will quickly return to it. However, if the high steady-state resource stock coincides with the tipping point then even though it is stable, a small perturbation can result in a large temporary fall in both recruitment and harvest rates. %
	
	Below the tipping point, a low-fecundity steady-state need not exist. First, the stationary point associated with the low-fecundity recruitment function may be above the tipping point, rendering it infeasible. Second, even when this stationary point is below the tipping point, if the fecundity differential between the high and low-fecundity recruitment functions is relatively small, the optimal harvest is austere and always leads to the high-fecundity steady-state.
	
	If the low-fecundity stationary point is feasible and the fecundity differential is sufficiently large then the instantaneous cost of austerity is relatively high. In this case, there is a second, endogenous threshold below the tipping point. If the initial resource stock is above this endogenous tipping point, the optimal harvest policy is austere and leads to the high-fecundity steady-state; even though the instantaneous cost of austerity is relatively high, the length of time this austerity must be borne is relatively low.  On the other hand, if the initial resource stock is below this endogenous tipping point then the length of time that austerity needs to be maintained is too high and instead, the optimal harvest policy is the standard one, leading to the low-fecundity steady-state.
	
	These results imply that when the initial resource stock is sufficiently high, the optimal harvest will always attain a high-fecundity stationary point, even from below the tipping point. But when the initial resource stock is small (due perhaps to over-harvesting), absent an external injection of the renewable resource, the optimal harvest policy does not attain high-fecundity. %
	
	When the renewable resource is subject to hysteresis, recruitment is history dependent. In particular, with hysteresis, at high-fecundity recruitment, there is a threshold resource stock below which the renewable resource transitions to low-fecundity and at low-fecundity, there is another, higher threshold required to transition back to high-fecundity \citep{scheffer:catastrophic, dudgeon:phase, selkoe:principles}. That is, at intermediate levels of the resource stock, both high and low-fecundity are possible and the current state of fecundity remains unchanged until the corresponding tipping point is crossed.
	On the high-fecundity recruitment function, if the resource stock falls below the high-fecundity tipping point, recruitment switches to low-fecundity. On the low-fecundity recruitment function, recruitment can only return to high-fecundity if the resource stock rises to the higher, low-fecundity tipping point.
	
	With hysteresis, if the high stationary point coincides with the high-fecundity tipping point then it is no longer stable. A small perturbation can bring the resource stock below the high-fecundity tipping point to low-fecundity recruitment. But now instead of quickly returning to the high-fecundity stationary point, there is either i) significant delay for the resource stock to rise to the (higher) low-fecundity tipping point or ii) recovery is not optimal.
	
	Recruitment with hysteresis accords with what we know about the Atlantic northwest cod fishery collapse of the early 1990s. While the population recovered modestly between 2005 and 2016, populations have since plateaued \citep{dfo:stock-2020,dfo:stock-2023} and the hoped for 2025 recovery \citep{rose:northern} has not come to fruition \citep{rose:northern2}.

	In the following Section, I review some of the related literature. In \Cref{sec:model}, I describe the model and in \Cref{sec:solution}, I examine the optimal renewable resource extraction problem without hysteresis. Then, in \Cref{sec:hysteresis}, I discuss how these results change when there is hysteresis in recruitment. Finally, in \Cref{sec:conclude}, I offer some concluding remarks.

	\section{Related Literature}
	
	Much of the theoretical renewable natural resource literature assumes that tipping is irreversible \citep[see][for examples]{polasky:optimal,zeeuw:managing,nkuiya:stochastic}. In \cite{reed:optimal}, regime shift implies the total collapse of the ecosystem and assumes that social welfare is linear in the harvest rate. He shows that optimal harvest policy may be either precautionary or aggressive in comparison to a model without regime shift. \cite{polasky:optimal} also assume that social welfare is linear but rather than catastrophic collapse, regime shift instead reduces the fecundity of the natural resource. They show that instead of being ambiguous, with a change in system dynamics, the optimal harvest is always precautionary. \cite{ren:optimal} and \cite{zeeuw:managing} generalize the model to allow for non-linear social welfare and show that if the planner is relatively indifferent to large fluctuations in the harvest rate (high elasticity of intertemporal substitution) then optimal harvests may become aggressive.
	
	There is a smaller literature that allows for ecosystem recovery. In \cite{reed:optimal}, ecosystem recovery is randomly exogenous so there is no role for policy to encourage recovery. More recently, a literature has developed that allows for endogenous recovery \citep{baggio-fackler:optimal, kvamsdal:optimal, arvaniti:time}. \cite{baggio-fackler:optimal} numerically evaluate a discrete time hazard model where recruitment is stochastic and the current regime may or may not be observable. Like \cite{polasky:optimal}, potential regime shift induces precautionary behavior. \cite{kvamsdal:optimal} generalizes \cite{baggio-fackler:optimal} to non-linear social welfare and multiple regimes and finds that similar to \cite{ren:optimal} and \cite{zeeuw:managing}, optimal harvests may be either conservative or aggressive.
	
	Finally, \cite{arvaniti:time} departs from the prior literature in several ways.  First, rather than modeling regime shift as a random hazard, they assume a fixed tipping point (like the current paper) below which recruitment has low fecundity. Like \cite{ren:optimal}, \cite{zeeuw:managing} and \cite{kvamsdal:optimal}, they assume that social welfare is continuous, increasing and concave but in contrast, they assume that it has a finite upper-bound. Similar to \cite{baggio-fackler:optimal}, recruitment is stochastic, however, unlike \citeauthor{baggio-fackler:optimal}, they show that there is a positive probability of total, irreversible collapse. Moreover, they make the nonstandard assumption that the planner is ``present biased'' or that they use quasi-hyperbolic discounting; without the ability to commit to a future course of action, they model the planner as a sequence of ``selves'' where each self plays a game against all of its future selves. Because of this, they are examining a second best solution, rather than the first best solution studied in much of the prior literature. Their key result is the unambiguous absence of precautionary resource extraction. The existence of a threshold induces a precautionary motive but this is countered by present bias which reduces the future value of the resource stock and reduces the value of ``reinvestment.'' Overall, the latter two effects outweigh the former effect.
	
	\section{The Model}	\label{sec:model}
	
	In dynamic models of renewable resources, growth of the resource stock is governed by a recruitment function, $f(x(t))$, where $x(t)$ is the resource stock at time $t$. In the standard analysis, $f$ is assumed to be a continuous function. In contrast, my interest is in functions that take discrete upward jumps.
	
	\begin{figure}
		\caption{Tipping recruitment}\label{fig:recruit}
		
		\begin{center}
			\includegraphics[scale=.4]{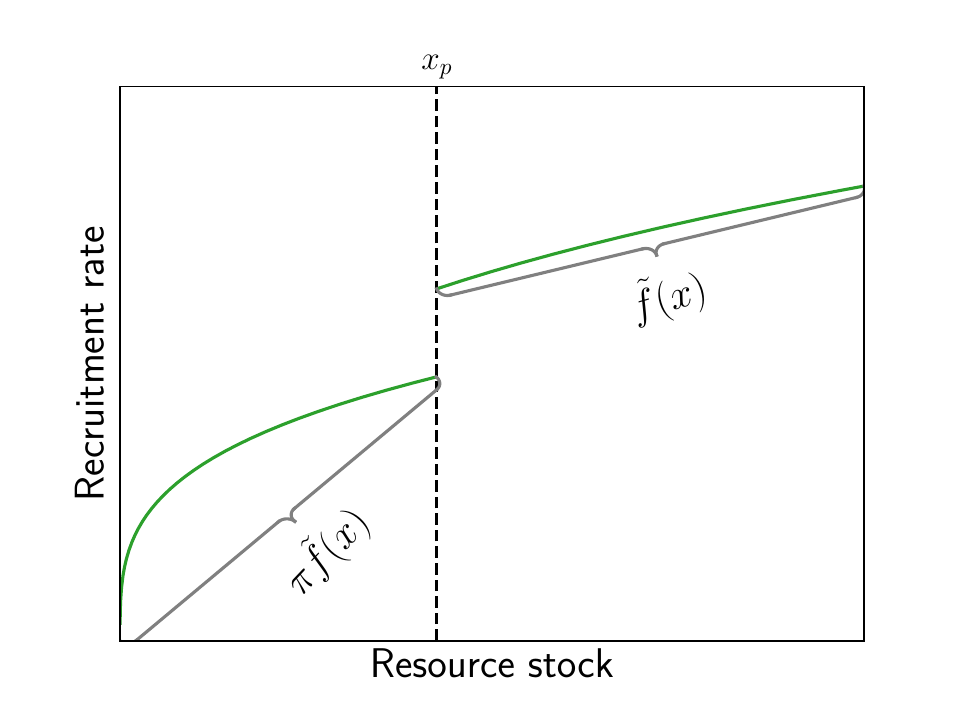}
		\end{center}
	\end{figure}
	
	\begin{figure*}
		\caption*{Legend for \Cref{fig:recruit,fig:no-tip,fig:constrained-upper,fig:optimal,fig:no-upper,fig:hyst-recruit,fig:hysteresis}}
		\begin{center}
			\includegraphics{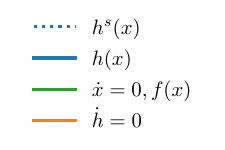}
		\end{center}
	\end{figure*}

	In particular, I define the tipping recruitment function as:
	\begin{equation}	\label{eq:tipping-recruit}
		f(x)=\begin{cases}
			\pi\tilde f(x) & \text{if }x<x_p \\
			\tilde f(x) & \text{if }x\ge x_p
		\end{cases}
	\end{equation}
	where $f(x)$ is the natural growth rate of the resource stock for $x\ge0$. The tipping point is $x_p>0$ and the tipping penalty is $0<\pi<1$. Below $x_p$, recruitment has low fecundity (the lower portion of Figure~\ref{fig:recruit}) and above, recruitment has high fecundity (the upper portion of Figure~\ref{fig:recruit}). The function $\tilde f$ is strictly increasing, %
	twice differentiable, concave, $\tilde f(0)=0$, $\lim_{x\rightarrow0}\tilde f'(x)=\infty$ and $\lim_{x\rightarrow\infty}\tilde f'(x)=0$.

	At time $t$, if $x(t)$ is the resource stock and $h(t)$ is the harvest rate then the growth rate of the resource stock is:
	\begin{equation}	\label{eq:fish-growth}
		\dot x(t)=f(x(t))-h(t).
	\end{equation}
	Net resource growth, $\dot x(t)$, is the natural rate at which the resource grows, $f(x(t))$, less the harvest rate, $h(t)$. When the harvest rate is below the recruitment rate ($h(t)<f(x(t))$), the resource stock is rising and when the harvest rate exceeds the recruitment rate ($h(t)>f(x(t))$), the resource stock is falling. At every time $t$, the resource stock and the harvest rate must be non-negative so that $x(t),h(t)\ge0$.
	
	Given harvest rate, $h(t)$ for $t\ge0$, the discounted social welfare is given by:
	\begin{equation}	\label{eq:utility}
		\int_{0}^{\infty}e^{-\rho t}u(h(t))dt
	\end{equation}
	where $\rho>0$ is the social discount rate %
	and $u(h(t))$ is instantaneous social welfare. Let $u$ be CRRA:
	\begin{equation}
		u(h) = \begin{cases}
					\frac{h^{1-\sigma}}{1-\sigma} & \text{if }\sigma\ne1 \\
					\ln h							  & \text{if }\sigma=1
			\end{cases}
	\end{equation}
	where $\sigma>0$ is the coefficient of relative risk aversion and its inverse, $\sigma^{-1}$, is the elasticity of intertemporal substitution.

	An optimal harvest plan solves:
	\begin{equation}	\label{eq:general_problem}
		\begin{gathered}
			V(x_0)=\max_{h(t)\ge0}\int_{0}^\infty e^{-\rho t}u(h(t))dt \\
			\text{s.t.\ }\dot{x}(t)=f(x(t))-h(t) \\
			x(t)\ge0 \\
			\text{given }x(0)=x_0>0
		\end{gathered}
	\end{equation}
	where $x_0$ is the initial resource stock. The analysis of this otherwise standard problem is complicated by the discontinuity in $f$.
	
	The current value Hamiltonian for this problem is:
	\begin{equation}	\label{eq:hamiltonian}
		\mathcal{H}(x,h,\lambda)=u(h)+\lambda[f(x)-h]
	\end{equation}
	where $\lambda$ is the costate and represents the value of an infinitesimal increase in the resource stock, $x$. I will proceed to the analysis of \eqref{eq:hamiltonian} in \Cref{sec:solution}. 
	In discussing trajectories (optimal and otherwise), it will be useful to refer to the policy function analogue, $h(x)$, %
	that dictates the harvest rate when the resource stock is $x$.
	
	\subsection{Smooth recruitment}	\label{sec:continous-problem}
	
	Before proceeding, I briefly review the analysis for the simpler case where there is no tipping point and the recruitment function is continuous. In particular, consider the recruitment function, $A\tilde f(x)$ where $A>0$. The solution to \eqref{eq:hamiltonian} satisfies the following necessary conditions:
	\begin{eqnarray}
		&h^{-\sigma}=\lambda,	\label{eq:foc} \\
		&\dot\lambda=\lambda[\rho-A\tilde f'(x)],	\label{eq:lambda-lom} \\
		&\dot x=A\tilde f(x)-h.	\label{eq:x-lom}
	\end{eqnarray}
	
	Differentiating \eqref{eq:foc} with respect to $t$ and using \eqref{eq:lambda-lom} yields:
	\begin{equation}	\label{eq:h-lom}
		\dot h=\frac{1}{\sigma}h[A\tilde f'(x)-\rho].
	\end{equation}
	This implies that the optimal harvest rate is increasing over time when the margin\-al recruitment rate exceeds the social discount rate ($A\tilde f'(x)>\rho$) and declining when the marginal recruitment rate is less than the social discount rate ($A\tilde f'(x)<\rho$). A transversality condition,
	\begin{equation}	\label{eq:transversality-standard}
		\lim\limits_{t\rightarrow\infty}e^{-\rho t}\lambda(t)x(t)=0
	\end{equation}
	implies that the discounted value of the resource stock is zero in the limit and ensures that harvests are optimal over the entire time horizon.
	
	\begin{figure}
		\caption{Smooth recruitment optimal harvest}	\label{fig:no-tip}
		
		\begin{center}
			\subfloat{\includegraphics[scale=.4]{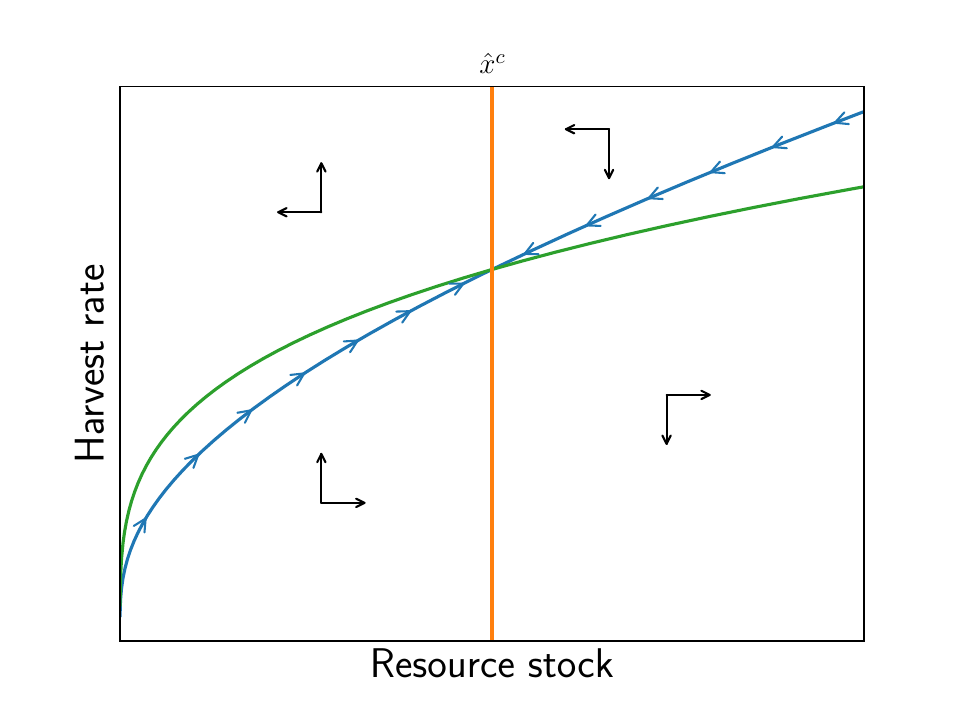}}
		\end{center}
	\end{figure}
	
	Together \eqref{eq:x-lom} and \eqref{eq:h-lom} represent an autonomous system of first-order differential equations that governs the model's dynamics.
	Since $\mathcal{H}$ is strictly concave in $(x,h)$, together with \eqref{eq:transversality-standard}, there is a unique solution, $(\hat x^c(t),\hat h^c(t))$, that converges to steady-state $(\hat x^c,\hat h^c)$ (\Cref{fig:no-tip}). Let $\hat h^c(x)$ and $V^c(x)$ denote the corresponding policy and value functions. The superscript $c$ denotes variables and functions associated with the continuous problem.
	
	\begin{remark*}
		Both here and subsequently, $(x(t),h(t))$ is a trajectory and $h(x)$ is a policy function. A ``hat'' denotes optimality and a ``hat'' variable with no functional argument is an optimal stationary point.
	\end{remark*}

	\subsection{Austerity}	\label{sec:austere}
	
	It will be useful for the analysis of tipping points to compare trajectories and policies to one another.  To do so, I define the notion of ``austerity.''
	
	Taking $A=\pi,1$, let $\underline f(\cdot)=\pi\tilde f(\cdot)$ and $\bar f(\cdot)=\tilde f(\cdot)$ be continuous  low- and high-fecundity, recruitment functions and denote the corresponding current value Hamiltonians as $\underline{\mathcal{H}}$ and $\overline{\mathcal{H}}$. Given $x_0$, let the unique, optimal trajectories be given by $(\hat{\underline x}(t),\hat{\underline h}(t))$ and $(\hat{\overline x}(t),\hat{\overline h}(t))$ for $t\ge0$. These trajectories converge to steady-states $(\hat{\underline x},\hat{\underline h})$ and $(\hat{\overline x},\hat{\overline h})$ and have corresponding policy functions $\hat{\underline h}(x)$ and $\hat{\overline h}(x)$ and value functions $\underline{V}(x)$ and $\overline{V}(x)$. Since $\pi<1$, it must be that $\hat{\underline h}(x)<\hat{\overline h}(x)$ and $\underline{V}(x)<\overline{V}(x)$. Call $(\hat{\underline x},\hat{\underline h})$ and $(\hat{\overline x},\hat{\overline h})$ the \textit{low} and \textit{high notional steady-states}.
	Define the ``standard'' policy function:
	\begin{equation}
		h^s(x)=\begin{cases}
			\hat{\underline h}(x) & \text{if }x<x_p \\
			\hat{\overline h}(x) & \text{if }x\ge x_p
		\end{cases}.
	\end{equation}
	
	I now define ``austerity.'' Loosely speaking, a trajectory $(x(t),h(t))$ is austere if it lies below the standard policy, $h^s(x)$.\footnote{
		In the renewable resource, hazard model literature, a reduced harvest policy is called ``precautionary'' because it reduces the likelihood of tipping \citep{polasky:optimal,zeeuw:managing}. In my model, since there is no uncertainty, reduced harvests cannot be ``precautionary'' and a more appropriate term is for harvests to be ``austere.'' Austere harvests can be employed to either prevent downward tipping or induce upward tipping. Unlike ``precautionary,'' ``austere'' is equally applicable to models with and without uncertainty.
	} To be precise:
	\begin{definition}
		Harvest policy $h(x)$ is \textbf{austere} relative to $h^0(x)$ if $h(x)\le h^0(x)$ and there is $x'<x''$ such that when $x\in [x',x'')$, $h(x)<h^0(x)$.
	\end{definition}
	\begin{definition}
		Trajectory $(x(t), h(t))$ is \textbf{austere} relative to harvest policy $h^0(x)$ if the corresponding harvest policy, $h(x)$, defined over the domain $[\inf\{x(t)\}_{t=0}^\infty,$ $\sup\{x(t)\}_{t=0}^\infty]$ is austere relative to $h^0(x)$ over the same domain.
	\end{definition}
	\begin{definition}
		A harvest policy $h(x)$ or a trajectory $(x(t),h(t))$ is \textbf{austere} if it is austere relative to $h^s(x)$.
	\end{definition}

	We will see that under the tipping recruitment function, $f(x)$, the optimal trajectory, $(\hat x(t),\hat h(t))$, may be austere.%

	\section{Optimal harvest}	\label{sec:solution}
	
	I solve the discontinuous renewable resource problem by construction.  I begin with the solution to the problem where $x_0\ge x_p$ and assume that the resource stock is constrained to remain at or above the tipping point. %
	The solution to this problem will yield a constrained optimal path, $(\hat x^*(t),\hat h^*(t))$, that converges to $(\hat x^*,\hat h^*)$.\footnote{
		The superscript *, here and subsequently, denotes variables and functions associated with the constrained, high-fecundity problem. Similarly, a subscript * denotes variables and functions associated with the subsequent low-fecundity problem.
	} For $x\ge x_p$, the corresponding harvest policy and value functions are $\hat h^*(x)$ and $V^*(x)$.
	
	Next, I solve the problem for $x_0<x_p$, allowing for the possibility that the optimal trajectory may transition to the high-fecundity recruitment function. For a given initial resource stock, $x_0$, there are two possible solutions: i)~the optimal trajectory, $(\hat x_*(t),\hat h_*(t))$, converges to the low notional stationary point, $(\hat x_*,\hat h_*)=(\hat{\underline x},\hat{\underline h})$ or ii) the optimal trajectory reaches the tipping point so that $\hat x_*(T)=x_p$ at time $T$ with terminal value $e^{-\rho T}V^*(x_p)$, under the assumption that harvests thereafter follow the constrained, high-fecundity solution that converges to $(\hat x^*,\hat h^*)$. %
	
	Finally, given these solutions, I show that when high-fecundity is sufficiently valuable, the constrained, high-fecundity solution is still optimal when $x(t)$ is unconstrained. Consequently, the solution to the low-fecundity problem is also optimal.
	
	\subsection{Constrained high-fecundity problem}	\label{sec:constrained}
	
	Consider the problem where the resource stock is constrained to stay at or above the tipping point:
	\begin{equation}	\label{eq:constrained-upper-problem}
		\begin{gathered}
			V^*(x_0)=\max_{h(t)\ge0}\int_{0}^\infty e^{-\rho t}u(h(t))dt \\
			\text{s.t.\ }\dot{x}(t)=\tilde f(x(t))-h(t) \\
			x(t)\ge x_p \\
			\text{given }x(0)=x_0\ge x_p.
		\end{gathered}
	\end{equation}
	This is problem \eqref{eq:general_problem} for $x_0\ge x_p$ where there is the additional constraint that $x(t)\ge x_p$ for all $t\ge0$.
	
	\begin{figure}
		\caption{Constrained high-fecundity problem}	\label{fig:constrained-upper}
		
		\begin{center}
			\subfloat[Interior\label{fig:constrained-upper-interior}]{\includegraphics[scale=.4]{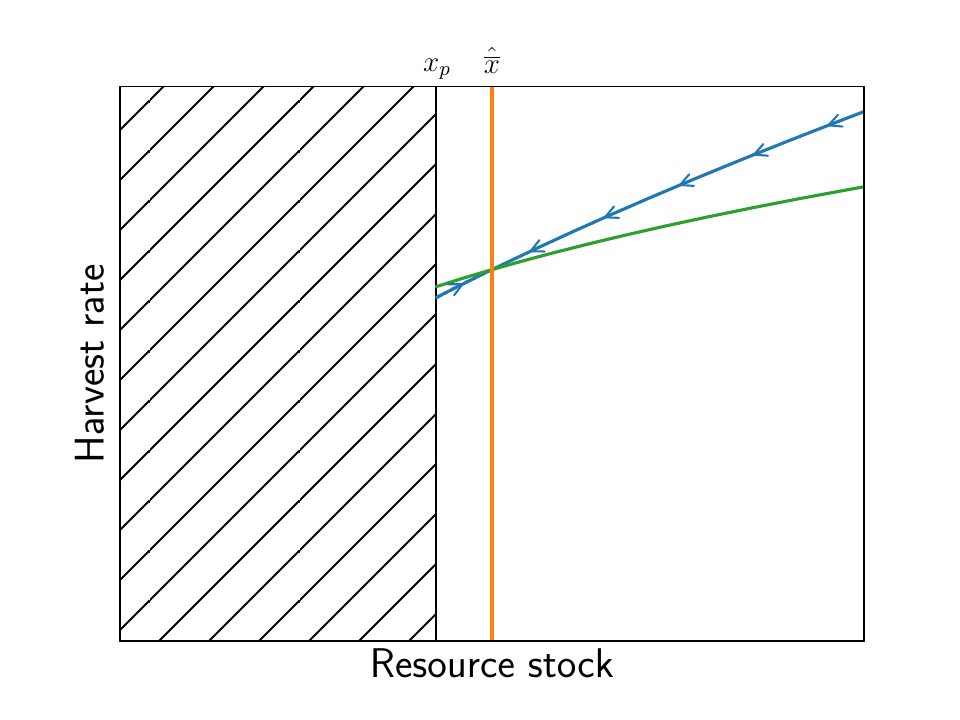}}~
			\subfloat[Boundary\label{fig:constrained-upper-boundary}]{\includegraphics[scale=.4]{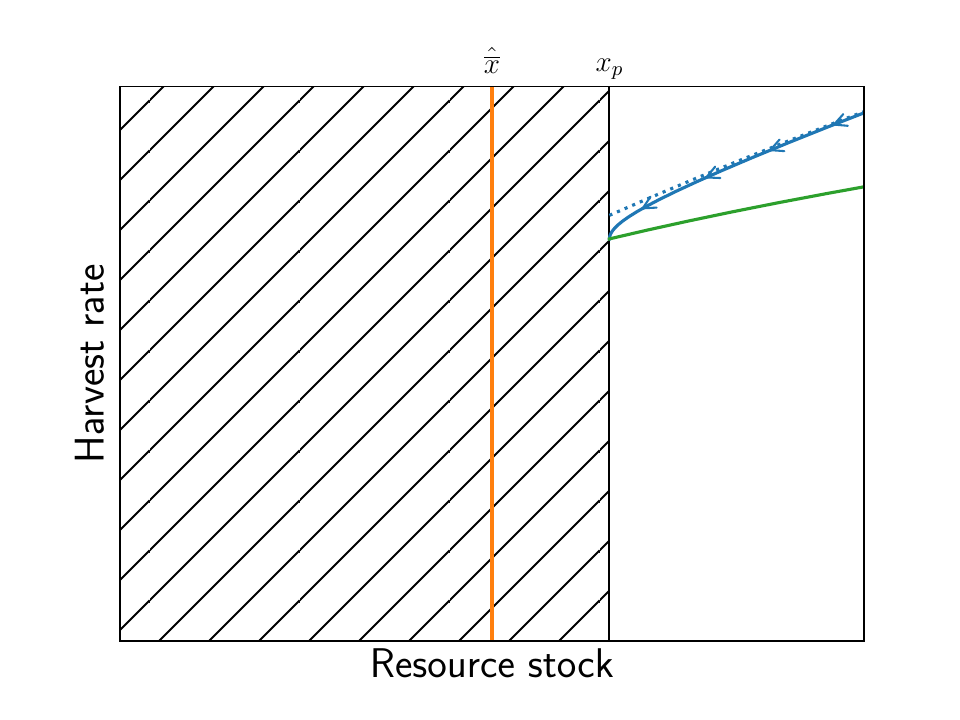}}
		\end{center}
	\end{figure}

	\begin{proposition}	\label{prop:high-fecund}
		For the high-fecundity problem given by \eqref{eq:constrained-upper-problem}, the optimal trajectory, $(\hat x^*(t),\hat h^*(t))$ for all $t\ge0$, is unique and:
		\begin{enumerate}[i)]
			\item if $\hat{\overline x}\ge x_p$ then $(\hat x^*(t),\hat h^*(t))=(\hat{\overline x}(t),\hat{\overline h}(t))$ and $\hat x^*=\hat{\overline x}$,
			\item if $\hat{\overline x}<x_p$ then $(\hat x^*(t),\hat h^*(t))$ is austere and there is some time $\tau<\infty$ such that for all $t\ge\tau$, $\hat x^*(t)=x_p$.
		\end{enumerate}
	\end{proposition}
	{\parindent0pt The proof of \Cref{prop:high-fecund} and all subsequent proofs can be found in Appendix~\ref{appendix:proofs}.}

	When the high notional stationary point is at least the tipping point ($\hat{\overline x}\ge x_p$), the constraint is non-binding and the optimal policy corresponds to the standard policy and the constrained optimal stationary point corresponds to the high notional stationary point defined in \Cref{sec:austere} (\Cref{fig:constrained-upper-interior}).
	
	But if the high notional stationary point is below the tipping point ($\hat{\overline x}<x_p$), the constraint is strictly binding and the optimal harvest is austere with the resource stock falling and stopping at $x_\tau=x_p$ at time $\tau$ (\Cref{fig:constrained-upper-boundary}). To see this, when the harvest trajectory is not sufficiently austere, the resource stock reaches the tipping point too quickly. On the other hand, when the harvest trajectory is too austere, the trajectory crosses the $\dot x=0$ line and is suboptimal since %
	harvesting $\tilde f(x_p)$ is always feasible.

	\subsection{Low-fecundity problem}	\label{sec:low-fecund}
	
	Now consider the low fecundity problem where $x_0<x_p$. Assume that if a trajectory, $(x(t),h(t))$, reaches the tipping point $x_p$ at some time $T$ then the planner gets ``terminal payoff,'' $e^{-\rho T}V^*(x_p)$. Beyond time $T$, the trajectory is assumed to follow the solution from \Cref{sec:constrained}.
	
	While it is always feasible for the resource stock to reach $x_p$, it may not be optimal. Thus there are two candidate outcomes. In one outcome, high-fecundity is not attained and the resource stock converges to the low notional steady-state, $\hat{\underline x}$. In the second outcome, the resource stock increases until it reaches the tipping point, $x_p$, whereupon recruitment becomes high-fecundity at time $T$.
	
	The optimization problem for the latter type of outcome is:
	\begin{equation}	\label{eq:low-fecund-free-T-problem}
		\begin{gathered}
			V_2(x_0)=\max_{h(t)\ge0}\int_0^T e^{-\rho t}u(h(t))dt+e^{-\rho T}V^*(x_p) \\
			\text{s.t. }\dot{x}(t)=\pi\tilde f(x(t))-h(t)\\
			x(t)\ge0 \\
			x(T)=x_p  \\
			T\text{ free} \\
			\text{given }x_0<x_p.
		\end{gathered}
	\end{equation}
	\sloppy This is a control problem with fixed terminal point, $x_p$, ``scrap value,'' $e^{-\rho T}V^*(x_p)$ and free terminal time, $T$.

	\fussy Either type of outcome must solve the current value Hamiltonian \eqref{eq:hamiltonian} so that both must satisfy \eqref{eq:x-lom} and \eqref{eq:h-lom} where $f(x)=\pi\tilde f(x)$ and $f'(x)=\pi\tilde f'(x)$. In addition, optimal trajectories must satisfy the appropriate transversality conditions. For the outcome that converges to the low notional stationary point, this is the standard transversality condition~\eqref{eq:transversality-standard}. For the outcome that transitions to high-fecundity at time $T$, the transversality condition is:
	\begin{equation}	\label{eq:fixed-terminal-value-transversality-lower}
		\lim_{t\rightarrow T}\underline{\mathcal{H}}(x(t),h(t),\lambda(t))=\rho V^*(x_p).
	\end{equation}
	This has the intuitive interpretation that as $t\rightarrow T$, the flow value of trajectory $(x(t),h(t))$, as represented by the current value Hamiltonian, must be equal to the flow value of the terminal payoff, $\rho V^*(x_p)$. Since $h(t)<\pi\tilde f(x(t))$, \Cref{lem:hamiltonian-slope} from the Appendix shows that $\underline{\mathcal H}(x,h,u'(h))$ is decreasing in $h$ and implies that if $\lim_{t\rightarrow T}\underline{\mathcal{H}}(x(t),h(t),\lambda(t))>\rho V^*(x_p)$ then high-fecundity is being attained too quickly so that $(x(t),h(t))$ is overly austere and larger harvests would be welfare improving. Conversely, if $\lim_{t\rightarrow T}\underline{\mathcal{H}}(x(t),h(t),\lambda(t))<\rho V^*(x_p)$ then the transition to high-fecundity is too slow and $(x(t),h(t))$ is insufficiently austere so that smaller harvests are optimal.
	
	The overall solution to the low-fecundity problem will have an endogenous tipping (Skiba) point, below which the standard trajectory obtains and above which an austere trajectory reaching $x_p$ and high fecundity is attained.
	
	\begin{proposition}	\label{prop:low-fecund}
		For the low-fecundity problem, if $\pi$, $\rho$ and $x_p$ are sufficiently small then the optimal trajectory, $(\hat x_*(t),\hat h_*(t))$ for all $t\ge0$, is unique and there exists $x_p'\in[0,x_p)$ such that
		\begin{enumerate}[i)]
			\item if $x_0<x_p'$ then  $(\hat x_*(t),\hat h_*(t))=(\hat{\underline x}(t),\hat{\underline h}(t))$ and $\hat x_*=\hat{\underline x}$,
			\item if $x_0\ge x_p'$ then $(\hat x_*(t),\hat h_*(t))$ is austere and there is some time $\tau<\infty$ such that $\hat x_*(\tau)=x_p$.
		\end{enumerate} 
	\end{proposition}
	
	When $\pi$, $\rho$ and $x_p$ are small, the discounted value of high-fecundity is relatively high and at the margin, the planner prefers high-fecundity to low-fecundity. If $\pi$, $\rho$ and $x_p$ are too large, the planner prefers low-fecundity and will not implement an austere harvest policy.

	If $\pi$, $\rho$ and $x_p$ are sufficiently low, then there is an endogenous tipping point, $x_p'$ (possibly trivial with $x_p'=0$). Above $x_p'$, the value of the terminal payoff, $e^{-\rho T}V^*(x_p)$, is relatively high and austere harvests can optimally attain high fecundity. Below $x_p'$, austerity is too costly and the planner employs the standard policy which converges to the low notional steady-state, $\hat{\underline x}$.
	
	\subsection{Unconstrained optimality}	\label{sec:unconstrained}
	
	Now consider the full problem where the resource stock is only constrained to be non-negative.	In particular, for $x_0\ge x_p$, an unconstrained trajectory can have $x(t)<x_p$ for some $t>0$. Let $(\hat x(t),\hat h(t))$ be the trajectory that solves this problem with corresponding policy function, $\hat h(x)$, and value function, $V(x)$.

	\begin{proposition}	\label{prop:full}
		For the unconstrained problem, if $\pi$, $\rho$ and $x_p$ are sufficiently small then the optimal trajectory, $(\hat x(t),\hat h(t))$ for all $t\ge0$, is unique and there exists $x_p'\in[0,x_p)$ such that
		\begin{enumerate}[i)]
			\item if $x_0\in(0,x_p')$ then the optimal trajectory is $(\hat x(t),\hat h(t))=(\hat{\underline x}(t),\hat{\underline h}(t))$ and $\hat x=\hat{\underline x}$,
			\item if $x_0\in[x_p',x_p)$ then the optimal trajectory $(\hat x(t),\hat h(t))$ is austere and there exists $\tau<\infty$ such that $\hat x_\tau=x_p$ and $\hat x=\max\{\hat{\overline x},x_p\}$,
			
			\item if $x_0,\hat{\overline x}\ge x_p$ then $(\hat x(t),\hat h(t))=(\hat{\overline x}(t),\hat{\overline h}(t))$ and $\hat x=\hat{\overline x}$,
			\item if $x_0\ge x_p>\hat{\overline x}$ then $(\hat x(t),\hat h(t))$ is austere and there exists $\tau<\infty$ such that $\hat x(t)=x_p$ for $t\ge\tau$.
		\end{enumerate}
	\end{proposition}
	
	\begin{figure}
		\caption{Examples with and without a high fecundity steady state}	\label{fig:no-upper}
		
		\subfloat[Large $\pi$, $x_p$ \label{fig:no-upper-pi}]{\includegraphics[scale=.4]{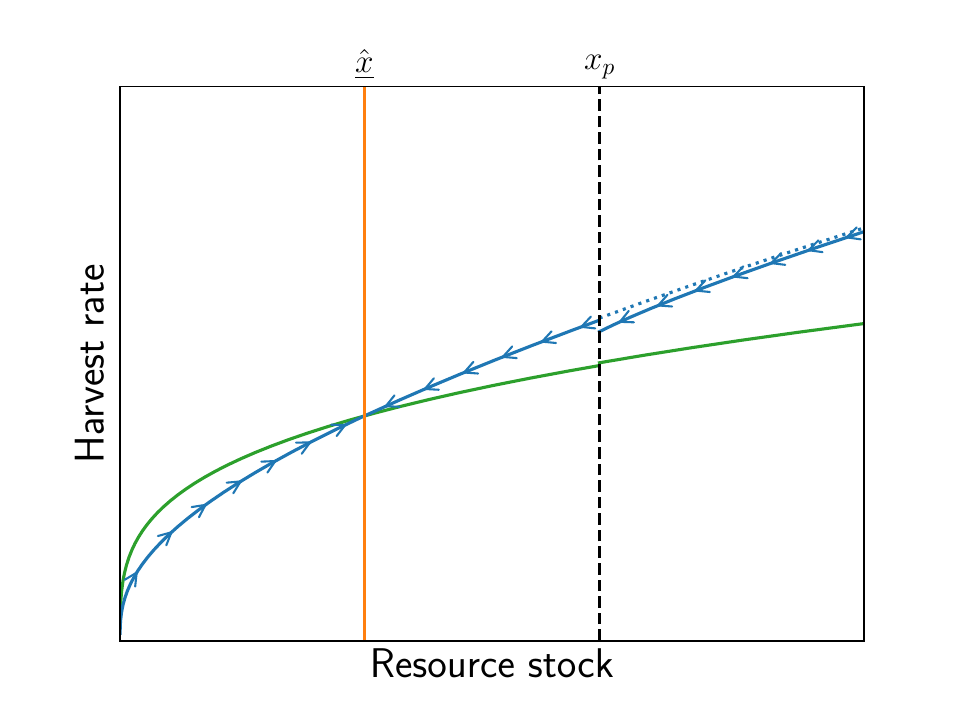}}~~
		\subfloat[Large $\rho$, $x_p$ \label{fig:no-upper-rho}]{\includegraphics[scale=.4]{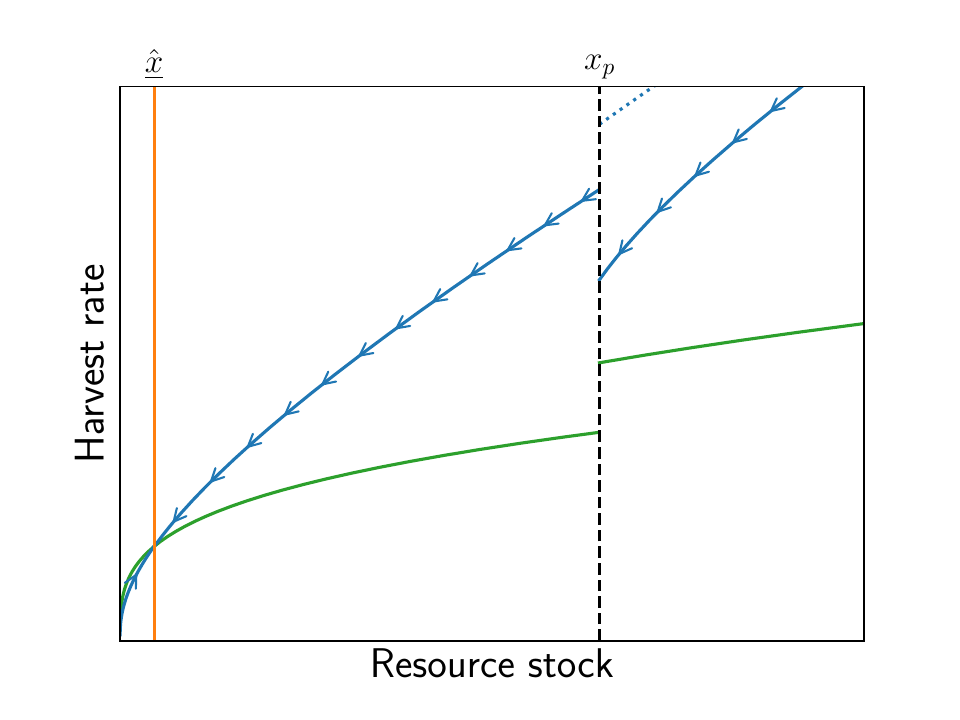}}
		
		\centering\subfloat[Large $\pi$, small $x_p$ \label{fig:upper}]{\includegraphics[scale=.4]{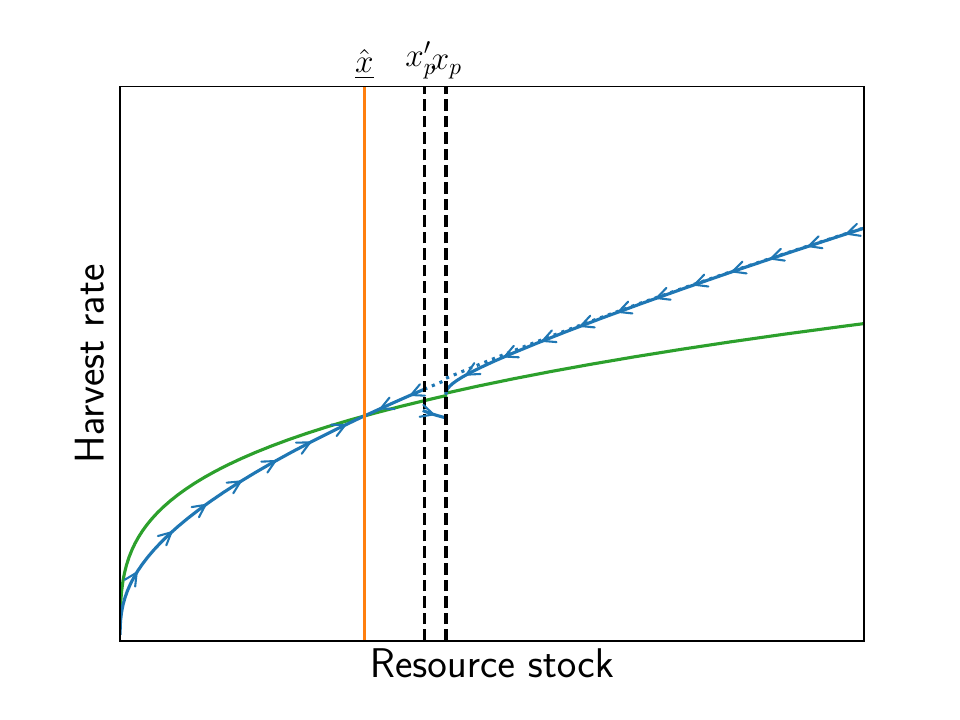}}
	\end{figure}
	
	When $\pi$, $\rho$ and $x_p$ are too large, there is no high-fecundity steady-state (\Cref{fig:no-upper-pi,fig:no-upper-rho}). Since the high notional steady-state is always preferable to the low notional steady-state, when these conditions fail, it must be the case that the high notional steady-state is infeasible and $\hat{\overline x}<x_p$. Examples of the failure of \Cref{prop:full} are illustrated in \Cref{fig:no-upper-pi} where $\pi$ is large and in \Cref{fig:no-upper-rho} where $\rho$ is large. Notice that at high-fecundity ($x_0>x_p$), optimal harvests are austere in order to prolong high-fecundity prior to the eventual transition to low-fecundity. Finally, if $x_p$ is not too large, the high notional steady-state can be restored (in \Cref{fig:upper}, $\pi$ and $\rho$ have the same values as in \Cref{fig:no-upper-pi} but $x_p$ is smaller). The formal analysis of cases \ref{fig:no-upper-pi} and \ref{fig:no-upper-rho} would first solve the low-fecundity problem under the assumption that the resource stock can never exceed the tipping point ($x(t)<x_p$). Then, taking this solution as given, solve the high-fecundity problem where $x_0\ge x_p$ where the planner may choose to tip the ecosystem at some time $T$, attaining continuation value $e^{-\rho T}V_*(x_p)$. Aside from this basic sketch, I do not formally analyze the case where there is no high-fecundity steady-state.

	When $\pi$, $\rho$ and $x_p$ are small, the discounted value of high-fecundity is high and the planner prefers high-fecundity to low-fecundity. Above the tipping point ($x_0\ge x_p$) there is a stable, high-fecundity steady-state. When $\hat{\overline x}>x_p$, the tipping point is strictly non-binding so that $(\hat x(t),\hat h(t))=(\hat{\overline x}(t),\hat{\overline h}(t))$ and the ecosystem converges to the notional, high-fecundity stationary point ($\hat x=\hat{\overline x}$) (\Cref{fig:low-low,fig:low-med,fig:low-high,fig:med-low,fig:med-med,fig:med-high}). But if $\hat{\overline x}<x_p$, then the high-fecundity steady-state occurs at the tipping point and $\hat x=x_p$. In order to reach this steady-state optimally, harvests must be austere; the standard harvest policy brings the resource stock to the tipping point too quickly. Even though the stationary resource stock $\hat x=x_p$ is stable, a small perturbation can result in a large fall in both harvests and recruitment (\Cref{fig:high-low,fig:high-med,fig:high-high}).
	
	\begin{figure}
		\caption{Optimal harvest policies} \label{fig:optimal}
		\centering 
		
		\subfloat[Low Tipping Point, Low $\pi$\label{fig:low-low}]{\includegraphics[scale=.4]{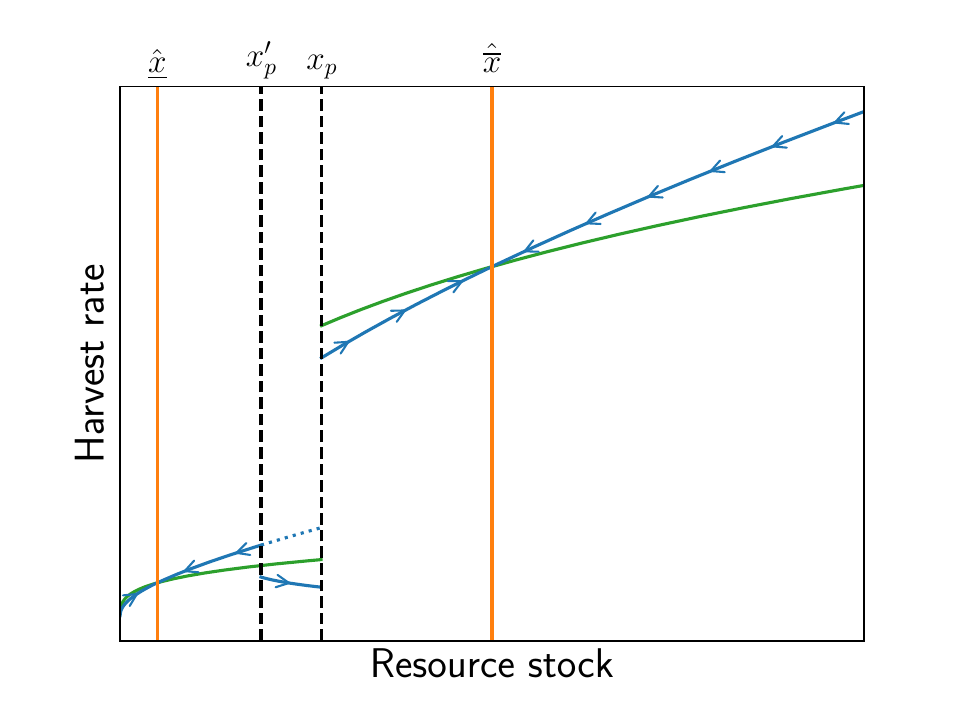}}~~
		\subfloat[Low Tipping Point, Med $\pi$\label{fig:low-med}]{\includegraphics[scale=.4]{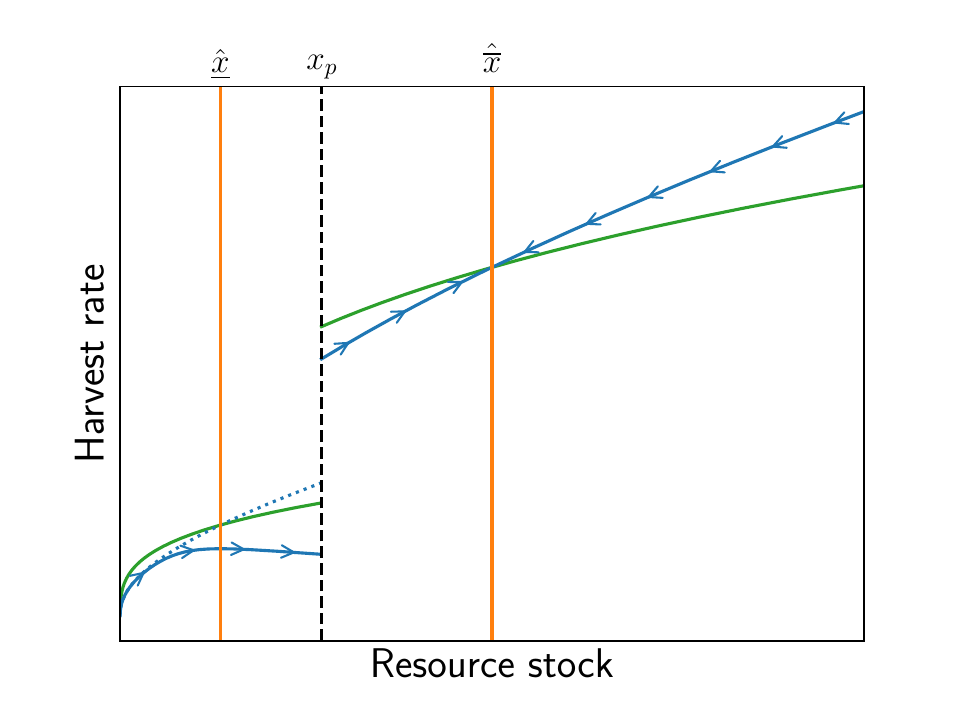}}
		
		\subfloat[Low Tipping Point, High $\pi$\label{fig:low-high}]{\includegraphics[scale=.4]{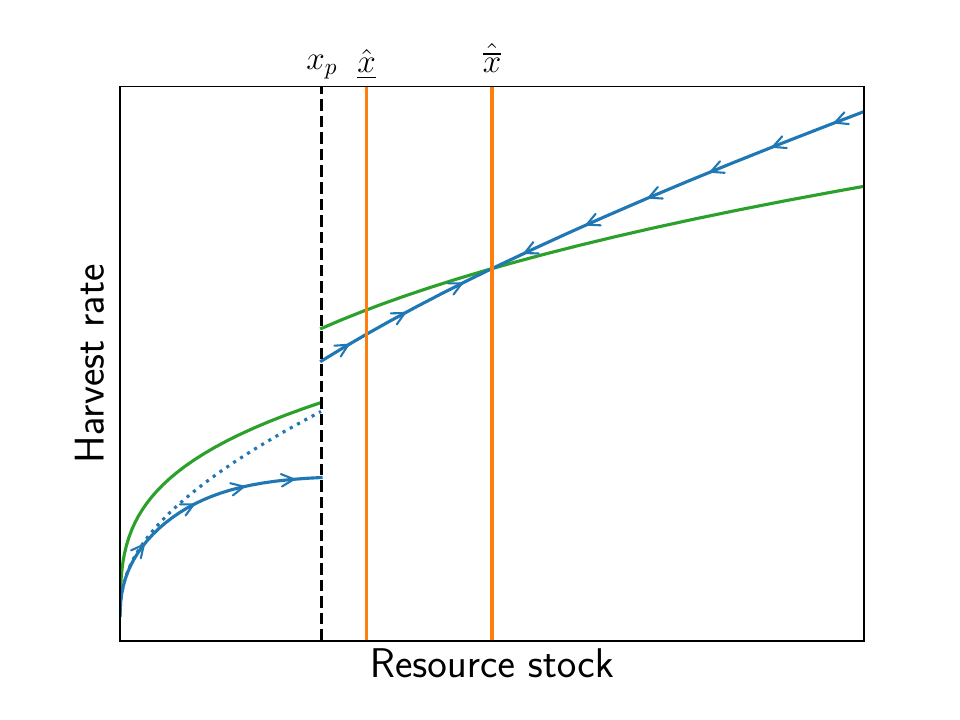}}~~			
		\subfloat[Med Tipping Point, Low $\pi$\label{fig:med-low}]{\includegraphics[scale=.4]{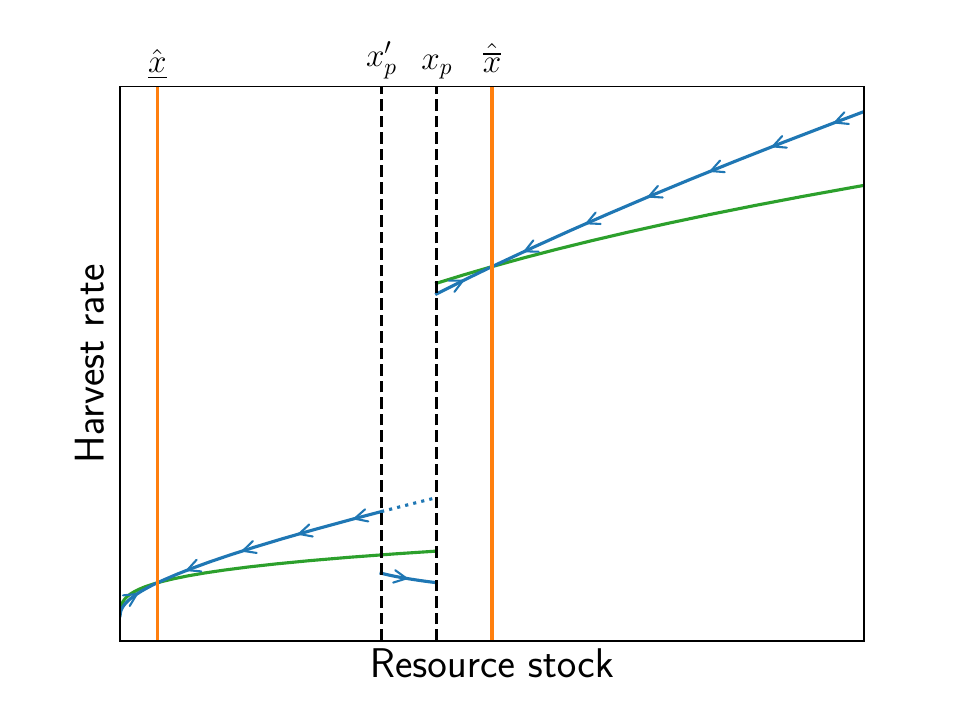}}
		
		\subfloat[Med Tipping Point, Med $\pi$\label{fig:med-med}]{\includegraphics[scale=.4]{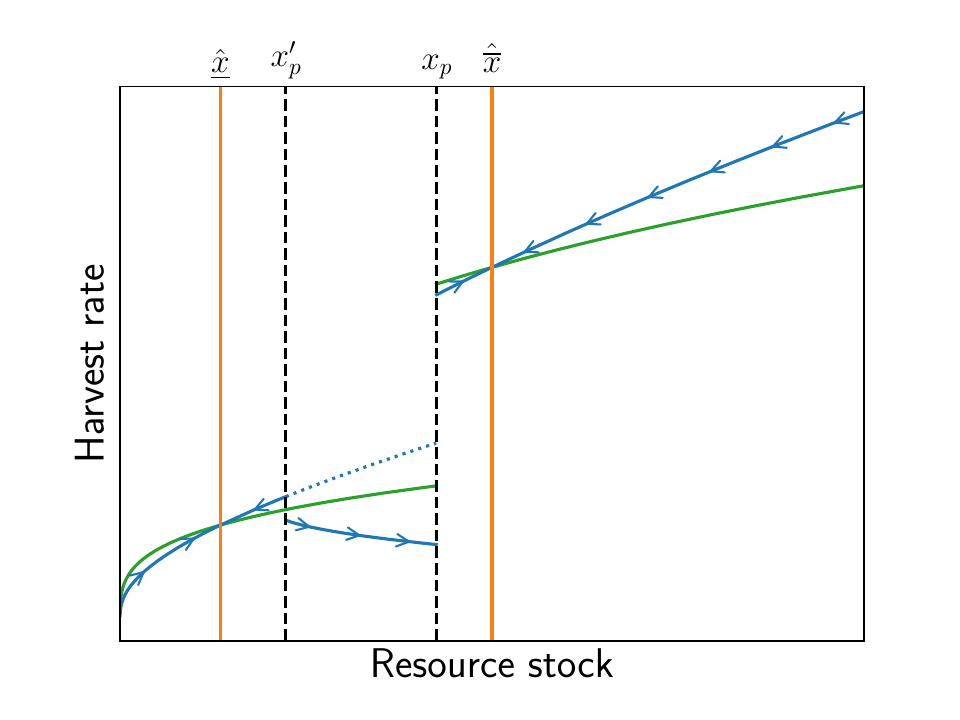}}~~
		\subfloat[Med Tipping Point, High $\pi$\label{fig:med-high}]{\includegraphics[scale=.4]{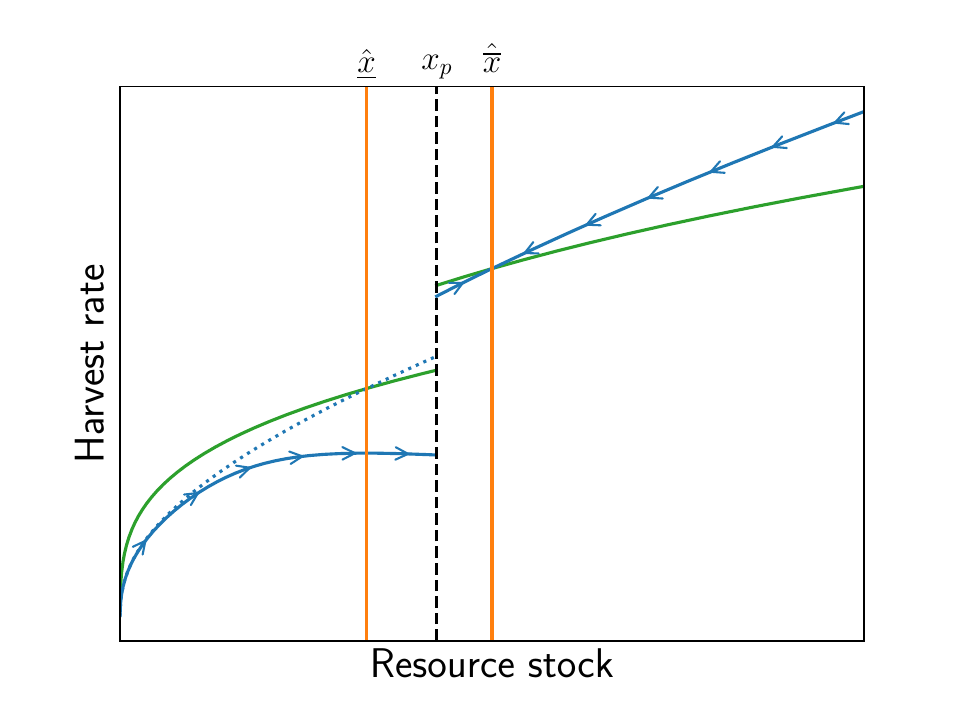}}
	\end{figure}	
	\begin{figure}
		\ContinuedFloat
		\caption{Optimal harvest policies (continued)}
		\centering 
		
		\subfloat[High Tipping Point, Low $\pi$\label{fig:high-low}]{\includegraphics[scale=.4]{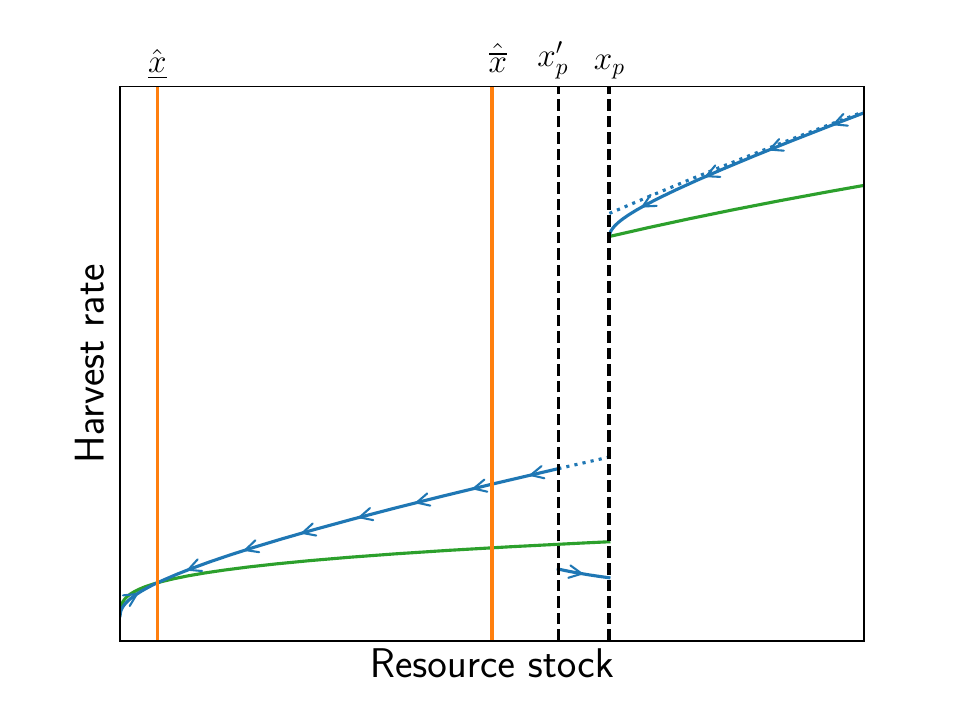}}~~
		\subfloat[High Tipping Point, Med $\pi$\label{fig:high-med}]{\includegraphics[scale=.4]{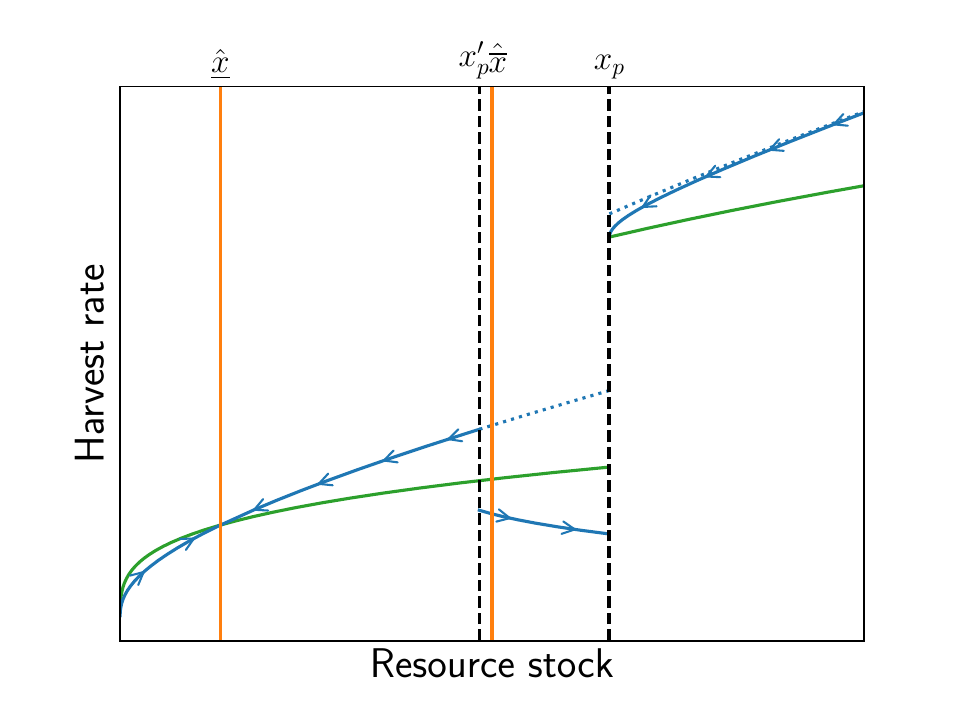}}
		
		\subfloat[High Tipping Point, High $\pi$\label{fig:high-high}]{\includegraphics[scale=.4]{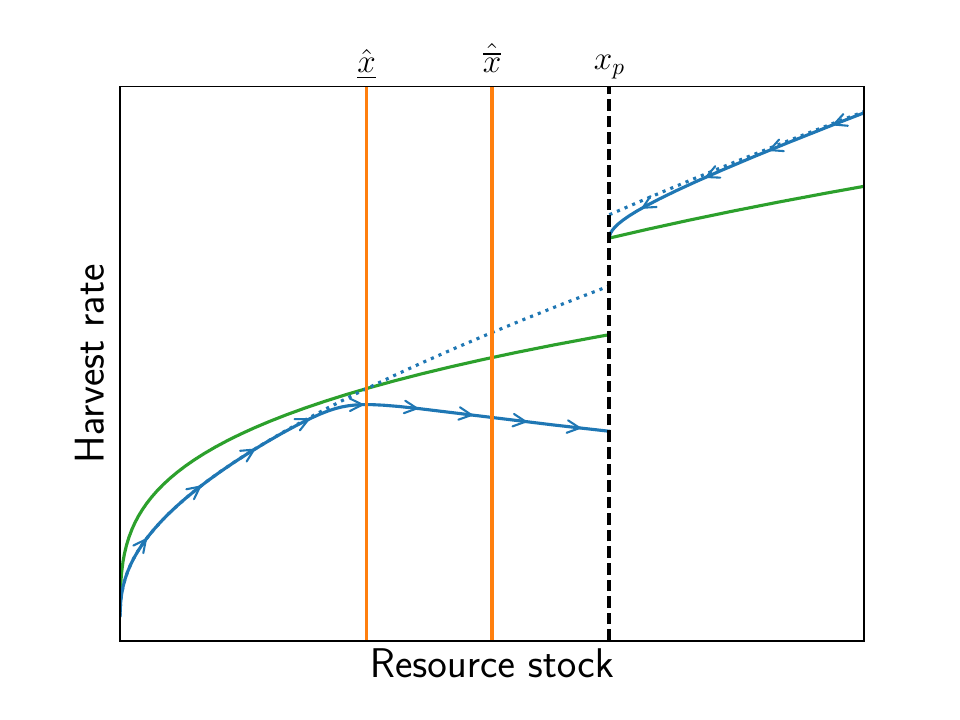}}
	\end{figure}
	
	Below the tipping point ($x_0<x_p$), there is an endogenous tipping point, $x_p'$.  For $x_0<x_p'$, the cost of austerity is too high and recovery, while feasible, is suboptimal; instead the resource stock converges to the low-fecundity stationary point and $\hat x=\hat{\underline x}$ (\Cref{fig:low-low,fig:med-low,fig:med-med,fig:high-low,fig:high-med}).  That is, given model parameters, below this endogenous tipping point, an austere harvest policy is inferior to the standard policy.
	For $x_0\ge x_p'$, austerity does not need to be borne for long and the optimal harvest transitions the renewable resource to high-fecundity and $\hat x\ge x_p$.
	
	Finally, it may be that the endogenous tipping point is inconsequential ($x_p'=0$).  When the low notional steady-state is infeasible ($\hat{\underline x}>x_p$) (\Cref{fig:low-high}) or when the high-/low-fecundity differential is relatively small (\Cref{fig:low-med,fig:med-high,fig:high-high}) then there is no non-trivial low steady-state. In these cases, as long as $x_0>0$, the optimal trajectory always reaches the high-fecundity steady-state and $\hat x=\max\{\hat{\overline x},x_p\}$. Despite the potential nonexistence of a bad long term outcome, as practitioners, we will be more interested in parameter configurations where the tipping point has a significant and tangible impact.
	
	\cite{baggio-fackler:optimal}, \cite{kvamsdal:optimal} and \cite{arvaniti:time} also model endogenous recovery. \citeauthor{baggio-fackler:optimal} has linear social welfare and finds austere harvests. \citeauthor{kvamsdal:optimal} has concave social welfare and instead finds either austere or aggressive harvests.  In \citeauthor{arvaniti:time}, harvests are always aggressive because quasi-hyperbolic discounting weights the present more than the future.
	
	Like \citeauthor{baggio-fackler:optimal}, my optimal harvest policy function is austere. However, my model features an endogenous tipping point, below which the discounted cost of austerity to be too high and low fecundity, while undesirable, is optimal. \citeauthor{arvaniti:time} also assume a fixed and known tipping point, however, quasi-hyperbolic discounting requires focusing on a second best scenario (the planner is unable to commit to a future course of action), rather than the first best scenario that is my own focus and that of much of the prior literature. Finally, as we will see in the following section, my framework is tractable enough to analytically model hysteresis in recruitment.

	\section{Hysteresis}	\label{sec:hysteresis}
	
	In recent years there is evidence that recovery from environmental damage can be subject to hysteresis (\citealt{field:coral, storlazzi:sedimentation} for coral reefs, \citealt{lindig-cisneros:wetland} for wetlands, \citealt{hirota:global} for rain forests). Hysteresis has become an important factor that marine ecologists consider as they seek to understand tipping points \citep{selkoe:principles}.
	
	\begin{figure}
		\caption{Hysteretic tipping recruitment}	\label{fig:hyst-recruit}
		
		\begin{center}
			\includegraphics[scale=.4]{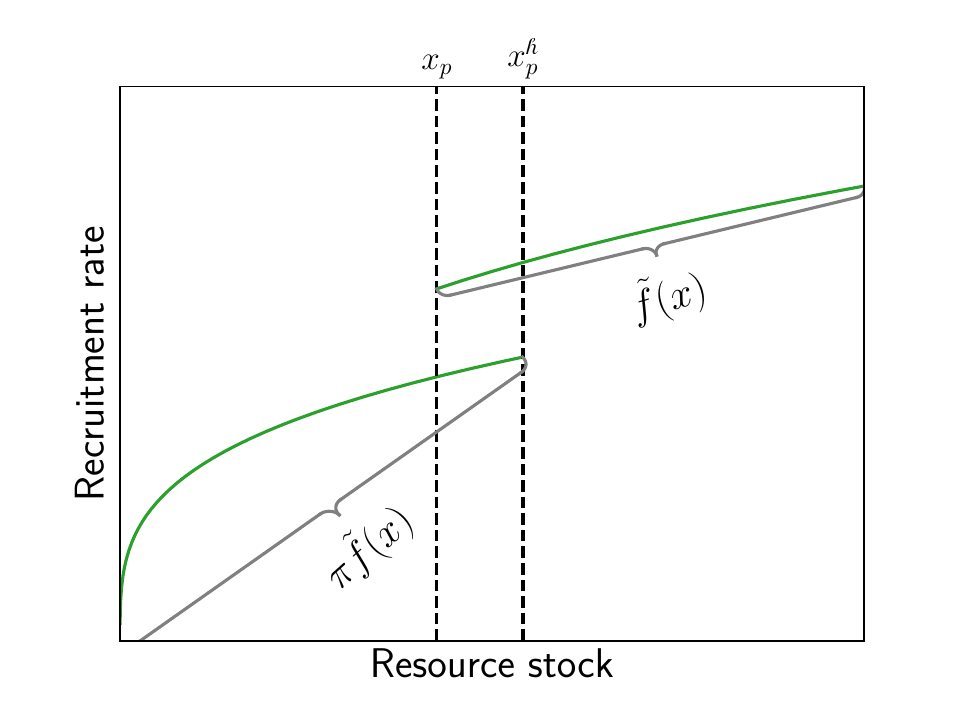}
		\end{center}
	\end{figure}
	
	When recruitment is subject to hysteresis, there are two tipping points. If fecundity is high then the tipping point is given by $x_p$. If the resource stock falls below this tipping point, the renewable resource switches to low-fecundity. With hysteresis, a higher tipping point must be reached in order for the renewable resource to transition to high-fecundity; the low-fecundity tipping point is given by $x_p^\mathcal{h}>x_p$. Functionally, the hysteretic recruitment function has a second argument, $s$:
	\[
		f(x,s)=
			(1-s)\pi\tilde f(x)+s\tilde f(x)
	\]
	where $s\in\{0,1\}$ is the ecosystem's state of fecundity with $s=1$ representing high-fecundity and $s=0$, low-fecundity. For $x<x_p$, $s=0$, for $x\ge x_p^\mathcal h$, $s=1$ and for $x\in[x_p,x_p^\mathcal h)$, $\dot s=0$ (i.e., $s$ retains its value). Recruitment can change discontinuously at $x_p$ and $x_p^\mathcal h$ (\Cref{fig:hyst-recruit}).

	In the model with hysteresis, denote the optimal trajectory $(\hat x^{\mathcal h}(t),\hat h^{\mathcal h}(t))$ with corresponding policy function, $\hat h^\mathcal h(x,s)$, and value function, $V^\mathcal h(x,s)$. With hysteresis,
	\begin{proposition}	\label{prop:hysteresis}
		If $\pi$, $\rho$ and $x_p^\mathcal h$ are sufficiently small then then the optimal trajectory, $(\hat x^{\mathcal h}(t),\hat h^{\mathcal h}(t))$ for $t\ge0$, is unique and there exists $x_p^{\mathcal{h}\prime}\in[0,x_p^\mathcal{h})$ such that:
		\begin{enumerate}[i)]
			\item if $x_0\in(0,x_p^{\mathcal{h}\prime})$ and $s_0=0$ then the optimal trajectory is $(\hat x^{\mathcal h}(t),\hat h^{\mathcal h}(t))=(\hat{\underline x}(t),\hat{\underline h}(t))$ and $\hat x^\mathcal h=\hat{\underline x}$,
			\item if $x_0\in[x_p^{\mathcal{h}\prime},x_p^\mathcal{h})$ and $s_0=0$ then the optimal trajectory $(\hat x^{\mathcal h}(t),\hat h^{\mathcal h}(t))$ is austere and there exists $\tau<\infty$ such that $\hat x_\tau^\mathcal h=x_p^\mathcal{h}$ and $\hat x^\mathcal h=\max\{\hat{\overline x},x_p\}$,
			
			\item if $x_0,\hat{\overline x}\ge x_p$ and $s_0=1$ then $(\hat x^{\mathcal h}(t),\hat h^{\mathcal h}(t))=(\hat{\overline x}(t),\hat{\overline h}(t))$ and $\hat x^\mathcal h=\hat{\overline x}$,
			\item if $x_0\ge x_p>\hat{\overline x}$ and $s_0=1$ then $(\hat x^{\mathcal h}(t),\hat h^{\mathcal h}(t))$ is austere and there exists $\tau<\infty$ such that $\hat x^{\mathcal h}(t)=x_p$ for $t\ge\tau$,
			\item $x_p'<x_p^{\mathcal h\prime}$.
		\end{enumerate}
	\end{proposition}

	As in the model without hysteresis, there is an endogenous tipping point. Above the endogenous tipping point, the cost of austerity is relatively low and the optimal harvest policy is austere and attains high-fecundity recruitment (\Cref{fig:hysteresis}). Since the low-fecundity tipping point is greater than the high-fecundity tipping point, upon reaching high-fecundity recruitment, the optimal harvest policy may then reverse course and spend down the resource stock to reach the high-fecundity steady-state (\Cref{fig:hyst-med-low,fig:hyst-med-med,fig:hyst-med-high,fig:hyst-high-low,fig:hyst-high-med,fig:hyst-high-high}).
	Below the endogenous tipping point, an austere harvest achieving high-fecundity recruitment is suboptimal and instead the optimal harvest follows the standard policy converging to the low-fecundity steady-state (\Cref{fig:hyst-low-low,fig:hyst-low-med,fig:hyst-med-low,fig:hyst-med-med,fig:hyst-high-low,fig:hyst-high-med,fig:hyst-high-high}).	
	
	\begin{figure}
		\caption{Hysteretic optimal harvest policies}	\label{fig:hysteresis}
		\centering
		
		\subfloat[Low Tipping Point, Low $\pi$\label{fig:hyst-low-low}]{\includegraphics[scale=.4]{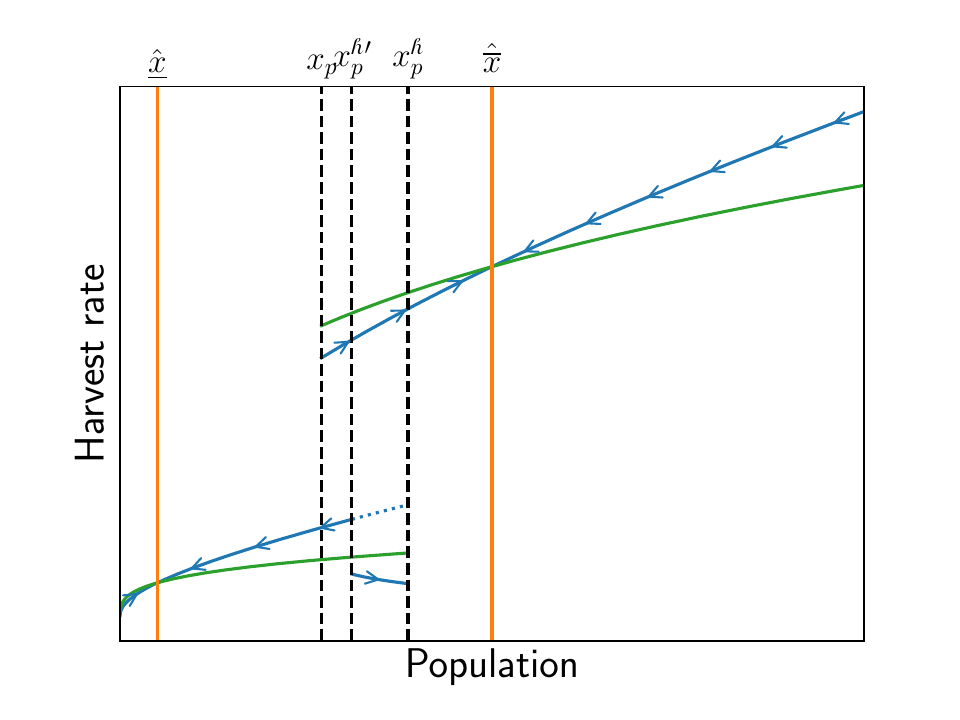}}~~
		\subfloat[Low Tipping Point, Med $\pi$\label{fig:hyst-low-med}]{\includegraphics[scale=.4]{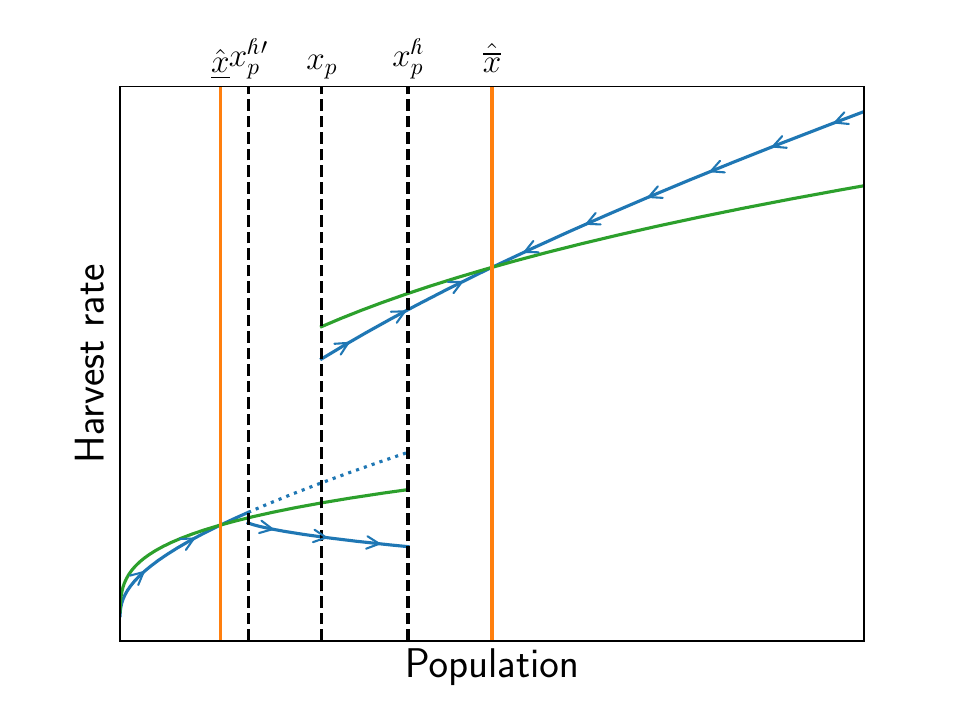}}
		
		\subfloat[Low Tipping Point, High $\pi$\label{fig:hyst-low-high}]{\includegraphics[scale=.4]{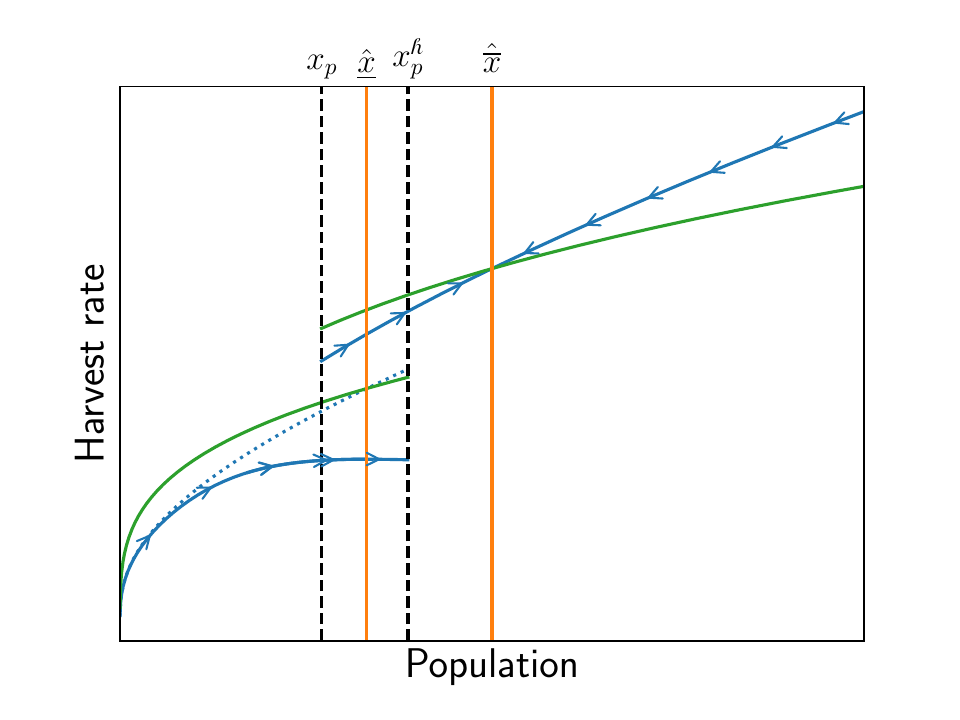}}~~			
		\subfloat[Med Tipping Point, Low $\pi$\label{fig:hyst-med-low}]{\includegraphics[scale=.4]{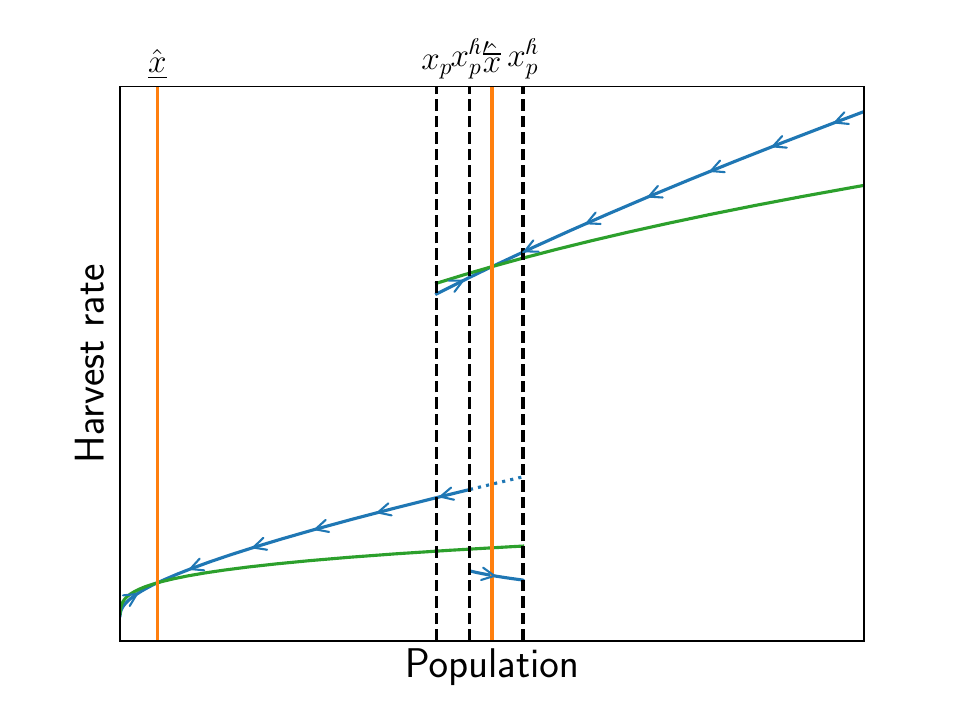}}
		
		\subfloat[Med Tipping Point, Med $\pi$\label{fig:hyst-med-med}]{\includegraphics[scale=.4]{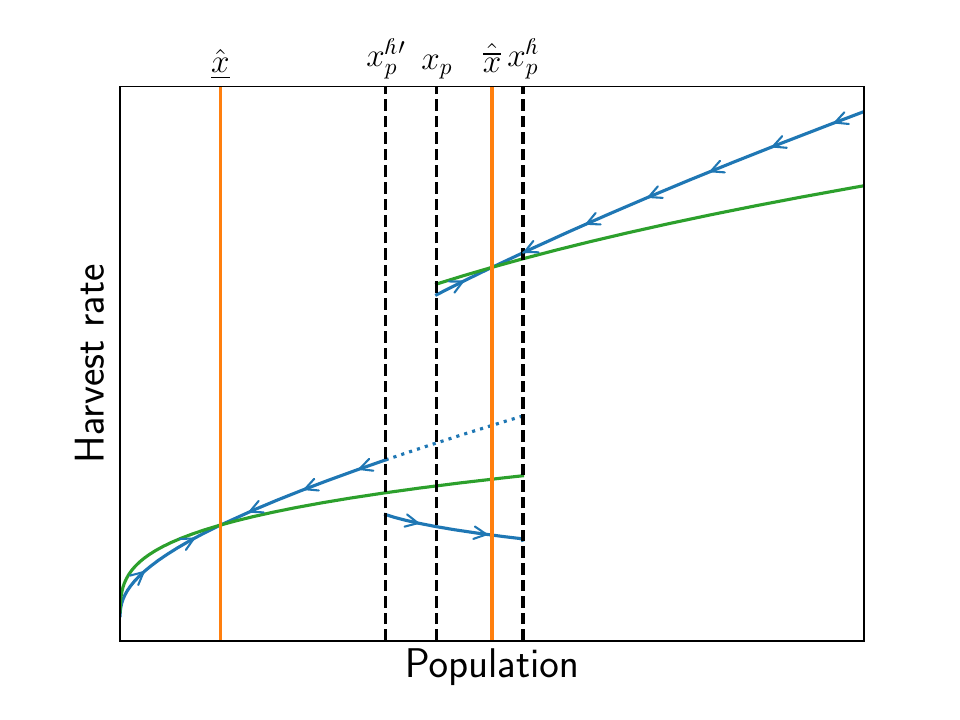}}~~
		\subfloat[Med Tipping Point, High $\pi$\label{fig:hyst-med-high}]{\includegraphics[scale=.4]{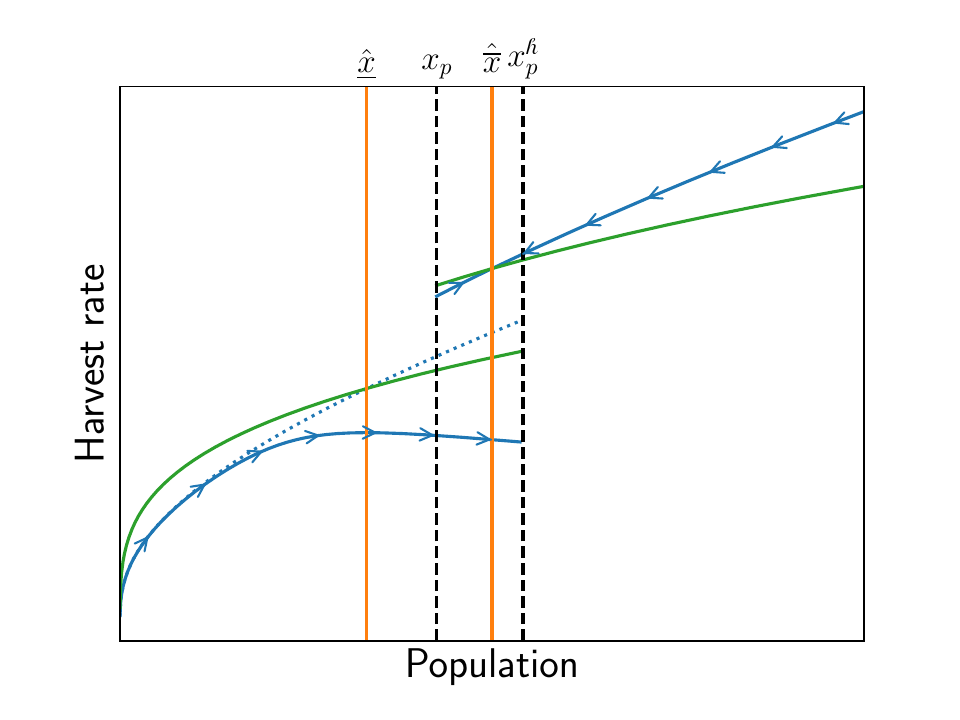}}				
	\end{figure}
	\begin{figure}
		\ContinuedFloat
		\caption{Hysteretic optimal harvest policies (continued)}
		\centering
		
		\subfloat[High Tipping Point, Low $\pi$\label{fig:hyst-high-low}]{\includegraphics[scale=.4]{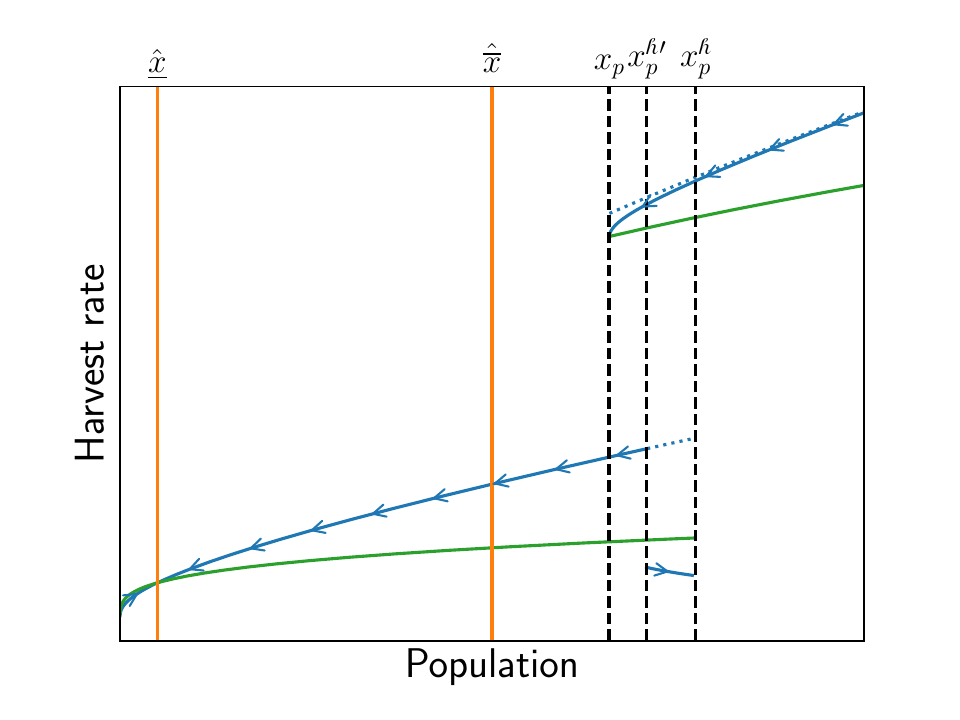}}~~
		\subfloat[High Tipping Point, Med $\pi$\label{fig:hyst-high-med}]{\includegraphics[scale=.4]{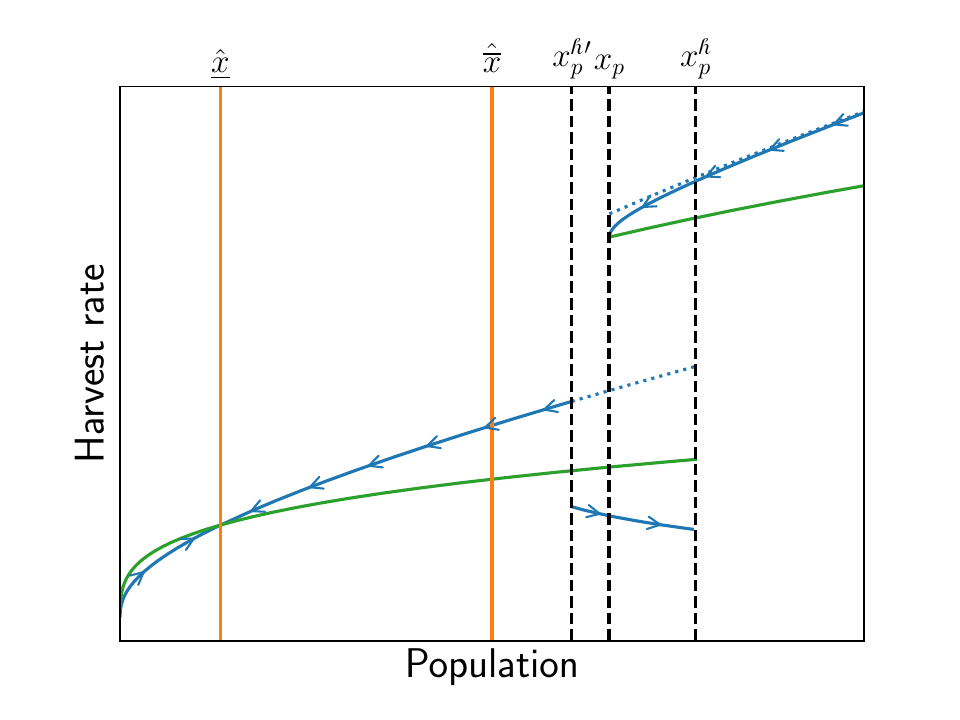}}
		
		\subfloat[High Tipping Point, High $\pi$\label{fig:hyst-high-high}]{\includegraphics[scale=.4]{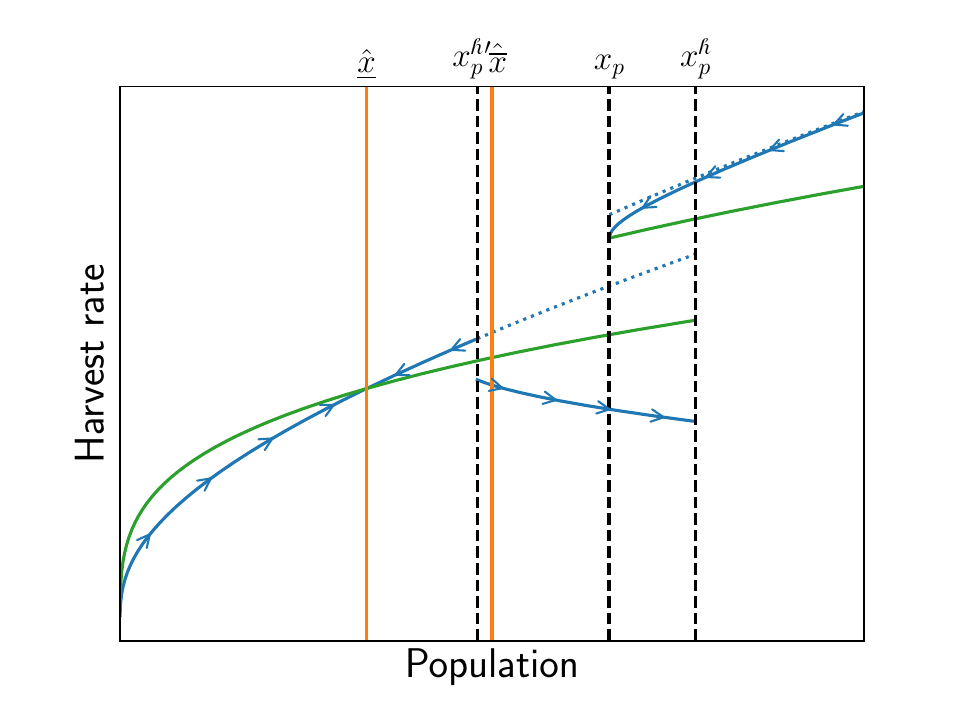}}
	\end{figure}
	
	In contrast to the model without hysteresis, a high-fecundity steady-state at the high-fecundity tipping point is no longer stable. If the high and low-fecundity differential is not too large then a perturbation that drops the resource stock below the high-fecundity tipping point requires an extended recovery period to return to high-fecundity, whereupon the optimal harvest spends down the resource stock to return to the steady-state (\Cref{fig:hyst-high-med,fig:hyst-high-high}). However, if the high- and low-fecundity differential is large then the endogenous tipping point, $x_p^{\mathcal h\prime}$, may be greater than the exogenous, high-fecundity tipping point. While returning to high-fecundity recruitment is feasible, it is suboptimal so that a fall below the tipping point becomes permanent (\Cref{fig:hyst-high-low}). That is, a high-fecundity stationary point that corresponds to the high-fecundity tipping point may not even be ``long run stable.''
	
	Finally, since $x_p^{\mathcal h\prime}>x_p'$ with hysteresis, the range over which initial resource stocks optimally remains at low-fecundity is larger.

	A disheartening example of the slow recovery from renewable resource collapse is the case of the Atlantic northwest cod fishery of the early 1990s \citep{hutchings:what, walters:lessons}. After more than three decades of restricted harvests, the Atlantic northwest cod fishery has still not recovered to sustainable levels \citep{dfo:stock-2020,dfo:stock-2023} and the hoped for 2025 recovery \citep{rose:northern} has not come to fruition \citep{rose:northern2}.

	\section{Conclusion}	\label{sec:conclude}
	
	In this paper I characterized the optimal extraction of a renewable resource where recruitment is subject to tipping points.  Much of the existing literature assumes that tipping is irreversible. In contrast, with a fixed tipping point, I am able to model renewable resource recovery. Moreover, when ecosystem recovery is possible, it becomes straightforward to model and analyze hysteresis. To the best of my knowledge, I am the first to model hysteresis.
	
	With endogenous reversibility, ecosystem recovery is not always optimal. When the ecosystem is sufficiently degraded, austerity becomes too costly and low fecundity becomes permanent. If in addition there is hysteresis, even a small perturbation that causes tipping may be permanent when the low-fecundity penalty is sufficiently severe. In this case, austerity may be too costly, even for an infinitesimal drop below the high-fecundity tipping point.
	
	This analysis presents a cautionary tale. First, it demonstrates that when an ecosystem suffers sufficient degradation, the resulting damage is irreversible. Moreover, with hysteresis, even a small perturbation can trigger long-lasting and potentially permanent changes. This underscores the delicate balance of natural systems and the critical importance of conservation efforts to prevent crossing ecological tipping points. %
	
	It is worth pointing out that this framework is applicable to other situations. Optimal renewable resource extraction models belong to the family of optimal growth models; recruitment is production and escapement (production less consumption) is investment. As such, the resource stock could be thought of as infrastructure, human or physical capital or any accumulated productive-resource.\footnote{
		\cite{azariadis:threshold} model ``threshold externalities'' in an overlapping generations framework. They motivate their article by illustrating their ideas using a simple example of capital accumulation where the production function takes the form of \Cref{fig:recruit}.
	} For instance, estimates of the cost of upgrading the entire U.S.\ rail system (freight and passenger) is estimated to be in the hundreds of billions and perhaps more than a trillion dollars.  Indeed, just to make overdue repairs to the passenger rail system is estimated to cost \$25.2 billion \citep{asce:comprehensive} and one estimate of upgrades amounts to \$209 billion \citep{narp:narp}. The modernization and electrification of the U.S.\ freight system may cost as much as \$1.1 trillion \citep{aar:study}. Due to inadequate past investment, the U.S.\ rail system may be below a critical threshold where the cost of repairs and upgrades are overwhelming.

	In this paper, I presented an ideal scenario where a social planner has complete control over harvests, all features of the ecosystem are known with certainty and there are no spillovers with other resources. With more than one resource extractor (tragedy of the commons), remaining above the tipping point would be more challenging \citep{levhari:great}. %
	Moreover, there may be important interactions between different renewable resources. For instance, deforestation results in the loss of habitat for native wildlife that may impact wildlife fecundity \citep{faria:breakdown}. Finally, uncertainty is important but has been left unmodeled. These complications, while important, are not considered here and call for further research.

	\appendix
	
	\section*{Appendix}
	
	\section{Proofs}	\label{appendix:proofs}
	
	\numberwithin{equation}{section}
	\setcounter{equation}{0}
	
	\begin{proof}[Proof of \Cref{prop:high-fecund}]
		
		\begin{enumerate}[i)]
			\item For $\hat{\overline x}\ge x_p$, since $(\hat{\overline x}(t),\hat{\overline h}(t))$ is optimal in the absence of constraint, it is also optimal with the constraint $x(t)\ge x_p$. Thus $(\hat x^*(t),\hat h^*(t))=(\hat{\overline x}(t),\hat{\overline h}(t))=(x^s(t),h^s(t))$ and $\hat x^*(t)\rightarrow\hat{\overline x}$.
		
			\item For $\hat{\overline x}<x_p$, there are two classes of trajectories $(x(t),h(t))$ satisfying \eqref{eq:x-lom} and \eqref{eq:h-lom}.
		
		In one, $(x(t),h(t))$ crosses the $\dot x=0$ curve (the green line in \cref{fig:constrained-upper-boundary}) at some point above $x_p$. Since $\hat{\overline x}<x_p$, the crossing point cannot be stationary ($\dot h(t)<0$) and standard arguments show that $(x(t),h(t))$ is suboptimal.
		
		In the second class, $(x(t),h(t))$ reaches $x(T)=x_p$ at some time $T$. For such trajectories, upon reaching $x(T)=x_p$, \eqref{eq:x-lom} and \eqref{eq:h-lom} imply that $(x(t),h(t))=(x_p,\tilde f(x_p))$ for $t>T$, otherwise the constraint that $x(t)\ge x_p$ would be violated. The optimal harvest problem can thus be rewritten as a free-terminal-time problem with terminal value, $e^{-\rho T}u(\tilde f(x_p))/\rho$:
		\[
		\begin{gathered}
			V^*(x_0)=\max_{h(t)\ge0}\int_0^Te^{-\rho t}u(x(t))dt+e^{-\rho T}\frac{u(\tilde f(x_p))}{\rho} \\
			\text{s.t. }\dot x(t)=\tilde f(x(t))-h(t) \\
			x(t)\ge x_p \\
			T \text{ free} \\
			\text{given }x_0\ge x_p
		\end{gathered}
		\]
		where the transversality condition is
		\begin{equation}	\label{eq:fixed-terminal-value-transversality-upper}
			\overline{\mathcal H}(x(T),h(T),\lambda(T))=u(\tilde f(x_p)).\footnote{
				See the discussion following \eqref{eq:fixed-terminal-value-transversality-lower} for the intuition behind this transversality condition.}
		\end{equation}
		But $u(h(T))+\lambda(T)[\tilde f(x_p)-h(T)]=u(\tilde f(x_p))$ if and only if $h(T)=\tilde f(x_p)$. Therefore, for $t<T$, $(\hat x^*(t),\hat h^*(t))$ is such that $\lim_{t\rightarrow T}\hat h^*(t)=\tilde f(x_p)$ and for $t\ge T$ $(\hat x^*(t),\hat h^*(t))=(x_p,\tilde f(x_p))$. Let $\tau=T$. Since $(\hat x^*(t),\hat h^*(t))$ is the only trajectory satisfying \eqref{eq:fixed-terminal-value-transversality-upper}, the solution is unique.
		
		Now consider trajectory $(x^s(t),h^s(t))$ that corresponds to the standard policy $h^s(\cdot)$ for a given $x_0$ (dotted blue line in \Cref{fig:constrained-upper-boundary}). In the continuous, high-fecundity problem, $x^s(t)\rightarrow\hat{\overline x}$ from above so that at some time $T'$, $x^s_{T'}=x_p$ and $h_{T'}^s>\tilde f(x_p)$. Therefore, $(\hat x^*(t),\hat h^*(t))$ and is austere.
		\end{enumerate}
	\end{proof}
	
	\begin{lemma}	\label{lem:high-fecund-v-bound}
		For the constrained high-fecundity problem, $V^*(x)\ge\frac{u(\tilde f(x_p))}{\rho}$.
	\end{lemma}
	\begin{proof}
		Since $\tilde f$ is strictly increasing, if $x(t)\ge x_p$ for any $t\ge0$ then harvesting $\tilde f(x_p)$ is always feasible but not necessarily optimal. Therefore
		\begin{equation*}
			\begin{aligned}
				V^*(x_0) & = \int_{0}^{\infty}e^{-\rho t}u(\hat h^*(t))dt \\
				&\ge\int_{0}^{\infty}e^{-\rho t}u(\tilde f(x_p))dt \\
				& =\frac{u(\tilde f(x_p))}{\rho}.
			\end{aligned}
		\end{equation*}
	\end{proof}
	
	\begin{lemma}	\label{lem:hamiltonian-slope}
		\begin{equation*}
			\mathcal H(x,h,u'(h))=u(h)+u'(h)[A\tilde f(x)-h]
		\end{equation*}
		is strictly decreasing in $h$ when $h<A\tilde f(x)$ and strictly increasing in $h$ when $h>A\tilde f(x)$.
	\end{lemma}
	\begin{proof}
		Differentiating with respect to $h$,
		\[
		\frac{\partial\mathcal H(x,h,u'(h))}{\partial h}=u''(h)[A\tilde f(x)-h].
		\]
		Since $u''(h)<0$, this is negative when $h<A\tilde f(x)$ and positive when $h>A\tilde f(x)$.
	\end{proof}
	
	\begin{lemma}	\label{lem:bound-low-policy}
		If $\pi$, $\rho$ %
		and $x_p$ are sufficiently small then $\hat{\underline h}(x_p)\le\tilde f(x_p)$.
	\end{lemma}
	The condition that $\hat{\underline h}(x_p)\le\tilde f(x)$ is used in the subsequent proofs and when $\pi$, $\rho$ %
	and $x_p$ are sufficiently small, \Cref{lem:bound-low-policy} ensures that this is the case.
	
	\begin{proof}
		Consider the limiting case where $\pi=0$ so that below the tipping point the resource is not renewable (i.e., the cake eating problem). The optimal harvest policy for CRRA instantaneous social welfare when $\pi=0$ and $x<x_p$ is $\underline{\hat h}(x)=\rho\sigma^{-1}x$. Given $\rho,\sigma>0$, since $\lim_{x\rightarrow0}\tilde f'(x)=\infty$, for sufficiently small $x_p$, $\underline{\hat h}(x_p)\le\tilde f(x_p)$. Alternatively, given $x_p$, for $\rho$ sufficiently small, $\hat{\underline h}(x_p)\le\tilde f(x_p)$. Since $\underline{\hat h}(\cdot)$ is continuous in $\pi$ and $\rho$, if $\pi$, $\rho$ and $x_p$ are sufficiently small then $\underline{\hat h}(x_p)\le\tilde f(x_p)$.
	\end{proof}

	\begin{proof}[Proof of \Cref{prop:low-fecund}]
		The plan of the proof is to: 1.\ Characterize candidate optimal trajectories: a) $(x_1(t),h_1(t))$ such that $\lim_{t\rightarrow\infty}(x_1(t),h_1(t))=(\hat{\underline x},\hat{\underline h})$ and b) $(x_2(t),h_2(t))$ where there exists some $\tau>0$ such that $x_2(\tau)=x_p$,~2.~Show that $(x_2(t),h_2(t))$ is austere and 3.\ Show that there is an endogenous tipping point, $x_p'$, below which $(x_1(t),h_1(t))$ is optimal and above which $(x_2(t),h_2(t))$ is optimal.
		
		\begin{enumerate}
			\item \begin{enumerate}[a)]
				\item Begin by considering trajectories, $(x_1(t),h_1(t))$, that converge to the low notional steady-state, $(\hat{\underline x},\hat{\underline h})$.
				
				If $\hat{\underline x}<x_p$ then the standard analysis shows that for $x_0<x_p$, $(\hat{\underline x}(t),\hat{\underline h}(t))$ is the unique trajectory satisfying \eqref{eq:x-lom}, \eqref{eq:h-lom} and \eqref{eq:transversality-standard} and converges to $(\hat{\underline x},\hat{\underline h})$. Let $(x_1(t),h_1(t))=(\hat{\underline x}(t),\hat{\underline h}(t))$ and denote the corresponding value function $V_1(x)=\underline{V}(x)$.
				
				If $\hat{\underline x}\ge x_p$ then $\hat{\underline x}$ is not feasible and cannot be a steady-state. Therefore there is no trajectory satisfying \eqref{eq:x-lom}, \eqref{eq:h-lom} and \eqref{eq:transversality-standard} that converges to $(\underline{\hat x},\underline{\hat h}))$ (see \cref{fig:low-high}).
				
				\item Now consider a trajectory $(x(t),h(t))$ such that at some time $T$, $x(T)=x_p$. In order to reach $x_p$ from $x_0<x_p$, it must be that $\dot x(t)>0$. Equation \eqref{eq:x-lom} implies that $h(t)<\pi\tilde f(x(t))$ and consequently, $\lim_{t\rightarrow T}h(t)\le\pi\tilde f(x_p)$; let $h^-(T)=\lim_{t\rightarrow T}h(t)$. Upon reaching $x(T)=x_p$, the terminal payoff is $e^{-\rho T}V^*(x_p)$.
				
				Given the terminal payoff $e^{-\rho T}V^*(x_p)$, the optimal trajectory minimizes the time required to reach $x_p$ while not sacrificing too much through austere early harvests. This balance is captured by the free-stopping-time transversality condition, \eqref{eq:fixed-terminal-value-transversality-lower}.
				
				For $h=\pi\tilde f(x_p)$,
				\[
				\begin{aligned}
					\underline{\mathcal{H}}(x_p,h,\lambda) & =u(\pi\tilde f(x_p)) \\
					& <u(\tilde f(x_p))=\rho\frac{u(\tilde f(x_p))}{\rho} \\
					& \le\rho V^*(x_p). & \text{\Cref{lem:high-fecund-v-bound}}
				\end{aligned}
				\]
				Since, $\lim_{h\rightarrow0}h^{1-\sigma}/(1-\sigma)+\pi\tilde f(x_p)h^{-\sigma}=\infty$, $\lim_{h\rightarrow0}\underline{\mathcal{H}}(x_p,h,u'(h))=\infty$. %
				Continuity implies that there is at least one $h^-(T)$ such that \eqref{eq:fixed-terminal-value-transversality-lower} is satisfied.
				
				From \Cref{lem:hamiltonian-slope}, if $h<\pi\tilde f(x)$ then $\underline{\mathcal H}(x,h,u'(h))$ is strictly decreasing in $h$. Therefore there is a unique $h^-(T)\in(0,\pi\tilde f(x_p))$ such that \eqref{eq:x-lom}, \eqref{eq:h-lom} and \eqref{eq:fixed-terminal-value-transversality-lower} hold. Let $(x_2(t), h_2(t))$ be the trajectory satisfying \eqref{eq:x-lom}, \eqref{eq:h-lom}, \eqref{eq:fixed-terminal-value-transversality-lower} such that $x_2(t)=x_p$. Let $\tau=T$.
			\end{enumerate}
			
			\item To show that $(x_2(t),h_2(t))$ is austere, consider two cases: a) $\hat{\underline x}\ge x_p$ and b) $\hat{\underline x}<x_p$.
			
			\begin{enumerate}[a)]
				\item For $\hat{\underline x}\ge x_p$, we know that $\hat{\overline x}>x_p$ so that the standard analysis holds for the high-fecundity problem and $\hat h^*(x_p)=\hat{\overline h}(x_p)$ and $V^*(x_p)=\overline{V}(x_p)$. Since $\overline{V}(x)>\underline{V}(x)$ it follows that:
				\begin{equation}	\label{eq:hamil-ineq-xp}
					\begin{aligned}
						\underline{\mathcal{H}}(x_p,h_2^-(\tau),u'(h_2^-(\tau))) & =\rho V^*(x_p) & \text{transversality } \\
						& =\rho\overline{V}(x_p) & \\
						& >\rho\underline{V}(x_p) \\
						& =\underline{\mathcal{H}}(x_p,\hat{\underline h}(x_p),u'(\hat{\underline h}(x_p))) & \text{HJB}.
					\end{aligned}
				\end{equation}
				where $h_2^-(\tau)=\lim_{t\rightarrow\tau}h_2(t)$.
				It was shown above that $h_2^-(\tau)\le\pi\tilde f(x_p)$ and since $\hat{\underline x}\ge x_p$, the low notional stationary point is not feasible and $\hat{\underline h}(x_p)\le\pi\tilde f(x_p)$. Since $\underline{\mathcal{H}}(x_p,h,u'(h))$ is decreasing in $h$ when $h\le\pi\tilde f(x)$ (\Cref{lem:hamiltonian-slope}), it must be that $h_2^-(\tau)<\hat{\underline h}(x_p)$. Therefore, $(x_2(t),h_2(t))$ is austere.
				
				\item For $\hat{\underline x}<x_p$, if $x_0\in(\hat{\underline x},x_p)$ then the low-fecundity optimal resource stock, $\hat{\underline x}(t)$, converges to $\hat{\underline x}$ from above (see dotted blue trajectories from \cref{fig:low-low,fig:low-med,fig:med-low,fig:med-med,fig:med-high,fig:high-low,fig:high-med,fig:high-high}) so that $\hat{\underline h}(x)>\pi\tilde f(x)$ for $x\in(\hat{\underline x},x_p)$. But we know that $(x_2(t),h_2(t))$ leads to $x_2(\tau)=x_p$ and must have $h_2(t)<\pi\tilde f(x_2(t))$ (follows from $\dot x_2(t)>0$). Therefore, $\hat{\underline h}(x_2(t))>h_2(t)$ and $(x_2(t),h_2(t))$ is austere.
			\end{enumerate}
			
			\item There are two candidate optimal trajectories and it remains to be determined when $(x_1(t),h_1(t))$ is optimal and when $(x_2(t),h_2(t))$ is optimal.
			
			For any trajectory, $(x(t),h(t))$, satisfying \eqref{eq:x-lom} and \eqref{eq:h-lom}, take the ratio of \eqref{eq:h-lom} and \eqref{eq:x-lom} to get an expression representing the slope of the corresponding policy function:
			\begin{equation*}	\label{eq:policy}
				\frac{dh}{dx}=\frac{dh/dt}{dx/dt}=\frac{1}{\sigma}\frac{h[\pi\tilde f'(x)-\rho]}{\pi\tilde f(x)-h}.
			\end{equation*}
			Now, totally differentiate $\underline{\mathcal H}(x,h,\lambda)$ where $\lambda=u'(h)$ and $h$ is a policy function satisfying \eqref{eq:x-lom} and \eqref{eq:h-lom}:
			\begin{equation}	\label{eq:hamiltonian-slope}
				\begin{aligned}
					\frac{\partial\underline{\mathcal H}}{\partial x}+
					\frac{\partial\underline{\mathcal H}}{\partial h}\frac{dh}{dx}+\frac{\partial\underline{\mathcal H}}{\partial\lambda}\frac{d\lambda}{dh}\frac{dh}{dx}&=
					u'(h)\pi\tilde f'(x)+
					[u'(h)-\lambda]\frac{dh}{dx}+[\pi\tilde f(x)-h]u''(h)\frac{dh}{dx} \\
					&=u'(h)\pi\tilde f'(x)+
					u'(h)[\rho-\pi\tilde f'(x)] \\
					&=u'(h)\rho>0.
				\end{aligned}
			\end{equation}
			Austerity of $h_2(x)$ implies $u'(h_2(x))>u'(h_1(x))$ since $u$ is strictly concave. Equation \eqref{eq:hamiltonian-slope} evaluated at $h_2(x)$ is always greater than when it is evaluated at $h_1(x)$ and therefore \eqref{eq:hamiltonian} evaluated at $h_2(x)$ is steeper than when evaluated at $h_1(x)$. Together with the HJB equation, this implies that for $x>0$, value functions $V_1(x)$ and $V_2(x)$ cross at most once. If there is a non-zero crossing point, call it $x_p'$; otherwise set $x_p'=0$.

			The discounted payoff for trajectory $(x_2(t),h_2(t))$ is:
			\begin{equation}	\label{eq:v2}
				V_2(x_0)=\int_{0}^\tau e^{-\rho t}u(h_2(t))dt+e^{-\rho\tau}V^*(x_p).
			\end{equation}
			Using the principle of optimality, the discounted payoff for trajectory $(x_1(t),h_1(t))$ can be written as:
			\begin{equation}	\label{eq:v1}
				V_1(x_0)=\int_0^\tau e^{-\rho t}u(h_1(x_1(t)))dt+e^{-\rho\tau}V_1(x_1(\tau)).
			\end{equation}
			When \eqref{eq:v1} is greater than \eqref{eq:v2}, trajectory $(x_1(t),h_1(t))$ is optimal; when \eqref{eq:v1} is less than \eqref{eq:v2}, trajectory $(x_2(t),h_2(t))$ is optimal.
			
			Consider $x_0\in(x_p-\varepsilon,x_p)$ for $\varepsilon>0$. As $\varepsilon\rightarrow0$, it follows that $\tau\rightarrow 0$ and $x_{1\tau}\rightarrow x_p$ so that the first terms in each of these equations vanishes while the second terms converge to $V^*(x_p)$ and $\underline V(x_p)$.
			
			It follows that,
			\[
			\begin{aligned}
				V^*(x_p)&\ge\frac{u(\tilde f(x_p))}{\rho} & \text{\Cref{lem:high-fecund-v-bound}} \\
				&\ge\frac{u(\hat{\underline h}(x_p))}{\rho}  & \hat{\underline h}(x_p)\le\tilde f(x_p) \\
				&>\int_0^\infty e^{-\rho t}u(\hat{\underline h}(t))dt \\
				&=\underline{V}(x_p).
			\end{aligned}
			\]
			The third inequality follows because for $x_0>\hat{\underline x}$, trajectory $(\hat{\underline x}(t),\hat{\underline h}(t))$ has $\dot{\underline x}(t),\dot{\underline h}(t)<0$ and thus $\hat{\underline h}(x_p)>\hat{\underline h}(t)$.
			
			Therefore, for $x_0$ sufficiently close to $x_p$ (the required duration of austerity is sufficiently small), \eqref{eq:v2} is greater than \eqref{eq:v1} and trajectory $(\hat x_2(t),\hat h_2(t))$ is optimal. %
			
			Now consider $x'$ such that $\lim_{x\downarrow x'} h_2(x)=\pi\tilde f(x')$. Since $h_2$ is austere and $h_2(x)<\pi\tilde f(x)$, it must be that if $x'>0$ then $x'\ge\hat{\underline x}$ (see \Cref{fig:low-low,fig:med-low,fig:med-med,fig:high-low,fig:high-med}).
			
			When $x'>0$, let $h_2'=\pi\tilde f(x')$. The current value Hamiltonian evaluated at $(x',h_2')$ and $\lambda=u'(h_2')$ is:
			\begin{equation}	\label{eq:hamiltonian-w-foc-2}
				\underline{\mathcal{H}}(x',h_2',u'(h_2'))=u(\pi\tilde f(x')).
			\end{equation}
			The current value Hamiltonian evaluated at $(x',\hat h_1(x'))$ is:
			\begin{equation}	\label{eq:hamiltonian-w-foc-1}
				\underline{\mathcal{H}}(x',\hat h_1(x'),u'(\hat h_1(x')))=u(\hat h_1(x'))+u'(\hat h_1(x'))[\pi\tilde f(x')-\hat h_1(x')])
			\end{equation}
			Note that for $h\ge\pi\tilde f(x)$, \eqref{eq:hamiltonian} %
			is increasing in $h$ (\Cref{lem:hamiltonian-slope}). Since $h_1(x')\ge\pi\tilde f(x')$, it must be the case that at $x'$, \eqref{eq:hamiltonian-w-foc-2} is no greater than \eqref{eq:hamiltonian-w-foc-1}.
			
			Now recall that at $x=x_p$, $\underline{\mathcal{H}}(x_p,h_2^-(\tau),u'(h_2^-(\tau)))>\underline{\mathcal{H}}(x_p,\hat{\underline h}(x_p),u'(\hat{\underline h}(x_p)))$ \eqref{eq:hamil-ineq-xp}. When $x'>0$, we know that $x'\ge\hat{\underline x}$ and continuity implies that there exists $x_p'\in[x',x_p)$ such that $\underline{\mathcal H}(x_p',\hat h_2(x_p'),u(\hat h_2(x_p')))$ $=\underline{\mathcal H}(x_p',\hat h_1(x_p'),u'(\hat h_1(x_p')))$. When $x'=0$, %
			it must be that $V_2(x)>V_1(x)$ for any $x\in(0,x_p)$; in this case set $x_p'=0$.

			The HJB equation implies that for $x<x_p'$, $V_2(x)<V_1(x)$ and for $x>x_p'$, $V_2(x)>V_1(x)$. Therefore,
			\[
			V_*(x)=\begin{cases}
				V_1(x) & \text{if }x<x_p' \\
				V_2(x) & \text{if }x\ge x_p'
			\end{cases}
			\]
			and
			\[
			\hat h_*(x)=\begin{cases}
				\hat h_1(x) & \text{if }x<x_p' \\
				\hat h_2(x) & \text{if }x\ge x_p'.
			\end{cases}
			\]
		\end{enumerate}
	\end{proof}
	
	\begin{proof}[Proof of \Cref{prop:full}]
		Note that from \Cref{lem:bound-low-policy}, we know that when $\pi$, $\rho$ and $x_p$ are sufficiently small, $\hat{\underline h}(x_p)\le\tilde f(x_p)$.
		
		For $x_0<x_p$, we know from \Cref{sec:low-fecund} that trajectory $(\hat x_*(t),\hat h_*(t))$ is optimal provided that the continuation value at time $\tau$ is $e^{-\rho\tau}V^*(x_p)$.	
		Now consider the unconstrained problem for $x_0\ge x_p$. From \Cref{prop:high-fecund}, we know that trajectory $(\hat x^*(t),\hat h^*(t))$ is optimal when $x(t)$ is constrained from falling below $x_p$.
		
		If $\hat{\overline x}\ge x_p$ then the notional steady-state is feasible and the constraint is non-binding so that trajectory $(\hat x^*(t),\hat h^*(t))=(\hat{\overline x}(t),\hat{\overline h}(t))$ is optimal.
		
		If $\hat{\overline x}<x_p$ then the question is whether there is an alternative, unconstrained trajectory, $(x^a(t),h^a(t))$, that satisfies \eqref{eq:x-lom} and \eqref{eq:h-lom} and attains greater welfare, say $V^a(x_0)$. From the argument in the proof of \Cref{prop:high-fecund}, we know that if $h^a(\hat x^*(t))<\hat h^*(t)$ then it is suboptimal.
		Alternatively, if $h^a(\hat x^*(t))>\hat h^*(t)$ then $x^a(t)$ and $h^a(t)$ fall continuously until at some $T<\tau$, $x^a(T)=x_p$ and $h^a(T)>\tilde f(x_p)$ (see
		\cref{fig:constrained-upper-boundary}). If $x^a(t)=x_p$ and $h^a(t)=\tilde f(x_p)$ for all $t>T$ then the transversality condition \eqref{eq:fixed-terminal-value-transversality-upper} fails at time $T$ and $(x^a(t),h^a(t))$ is suboptimal. If instead $x^a(t)$ falls below $x_p$ then \eqref{eq:x-lom} and \eqref{eq:h-lom} imply that $\dot h^a(t)<0$ and either $\dot x^a(t)>0$ or $\dot x^a(t)<0$ for $t>T$; $\dot x^a(t)>0$ can be ruled out because $x^a(t)$ would immediately return to $x_p$. But if $\dot x^a(t)<0$ then the only trajectory satisfying \eqref{eq:x-lom}, \eqref{eq:h-lom} and \eqref{eq:transversality-standard} is $(\hat{\underline x}(t), \hat{\underline h}(t))$. Consider the value from trajectory $(\hat x^*(t),\hat h^*(t))$ evaluated at $x_0$:
		\[
		\begin{aligned}
			V^*(x_0) & =\int_0^T e^{-\rho t}u(\hat h^*(t))dt+\int_T^\tau e^{-\rho t}u(\hat h^*(t))dt+\int_\tau^\infty e^{-\rho t}u(\tilde f(x_p))dt \\
			& \ge\int_0^T e^{-\rho t}u(h^a(t))dt+\int_T^\tau e^{-\rho t}u(\tilde f(x_p))dt+\int_\tau^\infty e^{-\rho t}u(\tilde f(x_p))dt \\
			& >\int_0^T e^{-\rho t}u(h^a(t))dt+\int_T^\tau e^{-\rho t}u(\hat{\underline h}(t))dt+\int_\tau^\infty e^{-\rho t}u(\hat{\underline h}(t))dt \\
			& = V^a(x_0).
		\end{aligned}
		\]
		The first inequality follows from the fact that $(\hat x^*(t),\hat h^*(t))$ is constrained optimal (\Cref{sec:constrained}) and $h^a(t)$ for $t\in[0,T)$ and $\tilde f(x_p)$ for $t\in[T,\tau)$ are feasible constrained harvests. The second inequality holds whenever $\hat{\underline h}(x_p)\le\tilde f(x_p)$ because $\dot{\hat{\underline h}}(t)<0$ for $t>T$. Therefore $(\hat x^*(t),\hat h^*(t))$ is optimal for the unconstrained problem. Since $(\hat x^*(t),\hat h^*(t))$ is unconstrained optimal, $(\hat x_*(t),\hat h_*(t))$ is optimal when $x_0<x_p$.
	\end{proof}

	\begin{proof}[Proof of \Cref{prop:hysteresis}]
		Note that the argument from the proof of \Cref{lem:bound-low-policy} can be used to show that if $\pi$, $\rho$ and $x_p^\mathcal h$ are sufficiently small then $\hat{\underline h}(x_p)\le\tilde f(x_p)$ and $\hat{\underline h}(x_p^\mathcal h)\le\tilde f(x_p^\mathcal h)$.
		
		As in the case without hysteresis, the problem will be divided between the high-fecundity problem where $x_0\ge x_p$, $s=1$ and recruitment is given by $\tilde f(x)$ and the low-fecundity problem where $x_0<x_p^\mathcal{h}$, $s=0$ and recruitment is given by $\pi\tilde f(x)$.
		
		The analysis for the high-fecundity problem with hysteresis is identical to the analysis without hysteresis and the optimal solution has trajectory $(\hat x^{\mathcal h*}(T),\hat h^{\mathcal h*}(T))=(\hat x^*(t),\hat h^*(t))$ for $t\ge0$ and value $V^{\mathcal h*}(x)=V^*(x)$.
		
		For the low-fecundity problem with hysteresis, the analysis of trajectories that converge to the low notional steady-state is identical to the model without hysteresis so that $(x_1(t)^\mathcal h,h_1(t)^\mathcal h)=(x_1(t),h_1(t))$ and $V_1^\mathcal h(x)=V_1(x)=\underline V(x)$.
		
		The analysis for the trajectories that transition to high-fecundity recruitment is slightly different and now occurs at some time $T$ such that $x_2^\mathcal h(T)=x_p^\mathcal h>x_p$. In this case, the transversality condition is now:
		\begin{equation*}	\label{eq:fixed-terminal-value-transversality-hysteresis}
			\lim_{t\rightarrow T}\underline{\mathcal{H}}(x(t),h(t),\lambda(t))=\rho V^{\mathcal{h}*}(x_p^\mathcal{h}).
		\end{equation*}
		The proof instead requires $\hat{\underline h}(x_p^\mathcal h)\le\tilde f(x_p^\mathcal h)$ but is otherwise identical; let the optimal trajectory be given by $(x_2^\mathcal h(t),h_2^\mathcal h(t))$ with value $V_2^\mathcal h(x)$.
		
		Proof of the optimality of the composite trajectory is the same as for \Cref{prop:full}, requiring $\hat{\underline h}(x_p)\le\tilde f(x_p)$. Let $x_p^{\mathcal h\prime}$ be the hysteretic endogenous tipping point,
		\begin{equation*}
			V_*^\mathcal h(x) = \begin{cases}
				V_1^\mathcal h(x) & \text{if }x<x_p^{\mathcal h\prime} \\
				V_2^\mathcal h(x) & \text{if }x\ge x_p^{\mathcal h\prime}
			\end{cases}.
		\end{equation*}
		The hysteretic value function is thus:
		\begin{equation*}
			V^\mathcal h(x,s)=\begin{cases}
				V_*^\mathcal h(x) & \text{if }x<x_p^\mathcal h\text{ and }s=0 \\
				V^{\mathcal h*}(x) & \text{if }x\ge x_p^\mathcal h\text{ and }s=1
			\end{cases}	
		\end{equation*}

		Note that it must be the case that for $x_p^\mathcal h\le x<x_p$, $V_2(x)>V_2^\mathcal h(x)$ since without hysteresis, consumption path $\hat h_2(t)^\mathcal h$ is feasible but $\hat h_2(t)$ is uniquely optimal. Clearly, $V_1(x)=V_1^\mathcal h(x)$. Together, this implies that $x_p'<x_p^{\mathcal h\prime}$.
	\end{proof}
	
	\nocite{skiba:optimal}

	\bibliography{growth}

\begin{thebibliography}{45}
\newcommand{\enquote}[1]{``#1''}
\providecommand{\natexlab}[1]{#1}

\bibitem[{AAR(2025)}]{aar:study}
AAR, 2025, \enquote{Study of Catenary Electrification of the North American
  Class I Railroad Network,} Tech. rep., Association of American Railroads,
  Washington, DC.

\bibitem[{Arvaniti et~al.(2023)Arvaniti, Krishnamurthy and
  Crépin}]{arvaniti:time}
Arvaniti, M., C.~K.~B. Krishnamurthy and A.-S. Crépin, 2023,
  \enquote{Time-consistent Renewable Resource Management with Present Bias and
  Regime Shifts,} \emph{Journal of Economic Behavior and Organization}, 207:
  479--495.

\bibitem[{ASCE(2021)}]{asce:comprehensive}
ASCE, 2021, \enquote{A Comprehensive Assessment of America's Infrastructure,}
  Tech. rep., American Society of Civil Engineers, Reston, VA.

\bibitem[{Azariadis and Drazen(1990)}]{azariadis:threshold}
Azariadis, C. and A.~Drazen, 1990, \enquote{Threshold {E}xternalities in
  {E}conomic {D}evelopment,} \emph{Quarterly Journal of Economics}, 105(2):
  501–526.

\bibitem[{Baggio and Fackler(2016)}]{baggio-fackler:optimal}
Baggio, M. and P.~L. Fackler, 2016, \enquote{Optimal Management with Reversible
  Regime Shifts,} \emph{Journal of Economic Behavior and Organization}, 132:
  124--136.

\bibitem[{Cass(1965)}]{cass:optimum}
Cass, D., 1965, \enquote{Optimum Growth in an Aggregative Model of Capital
  Accumulation,} \emph{Review of Economic Studies}, 32(3): 233--240.

\bibitem[{Clark(1973{\natexlab{a}})}]{clark:economics}
Clark, C.~W., 1973{\natexlab{a}}, \enquote{The Economics of Overexploitation,}
  \emph{Science}, 181(4100): 630--634.

\bibitem[{Clark(1973{\natexlab{b}})}]{clark:profit}
Clark, C.~W., 1973{\natexlab{b}}, \enquote{Profit Maximization and the
  Extinction of Animal Species,} \emph{Journal of Political Economy}, 81(4):
  950--961.

\bibitem[{Clark(2010)}]{clark:mathematical}
Clark, C.~W., 2010, \emph{Mathematical {B}ioeconomics: {T}he {M}athematics of
  {C}onservation}, Hoboken, NJ: Wiley, 3rd edn.

\bibitem[{de~Zeeuw and He(2017)}]{zeeuw:managing}
de~Zeeuw, A. and X.~He, 2017, \enquote{Managing a Renewable Resource Facing the
  Risk of a Regime Shift in the Ecological System,} \emph{Resource and Energy
  Economics}, 48: 42--54.

\bibitem[{DFO(2021)}]{dfo:stock-2020}
DFO, 2021, \enquote{2020 Stock Status Update for Northern Cod,} Tech. rep., DFO
  Canadian Science Advisory Secretariat.

\bibitem[{DFO(2024)}]{dfo:stock-2023}
DFO, 2024, \enquote{Update of Stock Status Indicators for the Northern Gulf of
  St. Lawrence Atlantic Cod Stock (3PN, 4RS) in 2023,} Tech. rep., DFO Canadian
  Science Advisory Secretariat.

\bibitem[{Dudgeon et~al.(2010)}]{dudgeon:phase}
Dudgeon, S.~R. et~al., 2010, \enquote{Phase {S}hifts and {S}table {S}tates on
  {C}oral {R}eefs,} \emph{Marine Ecology Progress Series}, 413: 201--216.

\bibitem[{FAO and UNEP(2020)}]{fao:state}
FAO and UNEP, 2020, \enquote{The State of the World's Forests 2020: Forests,
  biodiversity and people,} resreport, Food and Agriculture Organization of the
  United Nations and United Nations Environmental Programme, Rome.

\bibitem[{Faria et~al.(2023)}]{faria:breakdown}
Faria, D. et~al., 2023, \enquote{The Breakdown of Ecosystem Functionality
  Driven by Deforestation in a Global Biodiversity Hotspot,} \emph{Biological
  Conservation}, 283: 110,126.

\bibitem[{Felbab-Brown(2017)}]{felbab-brown:extinction}
Felbab-Brown, V., 2017, \emph{The Extinction Market: Wildlife Trafficking and
  How to Counter It}, Oxford: Oxford University Press.

\bibitem[{Field et~al.(2007)}]{field:coral}
Field, M.~E. et~al., 2007, \enquote{The Coral Reef of South Moloka'i, Hawai'i:
  Portrait of a Sediment-Threatened Fringing Reef,} Tech. rep., U.S. Geological
  Survey.

\bibitem[{Gordon(1954)}]{gordon:economic}
Gordon, H.~S., 1954, \enquote{The Economic Theory of a Common-Property
  Resource: the Fishery,} \emph{Journal of Political Economy}, 62(2): 124--142.

\bibitem[{Hirota et~al.(2011)}]{hirota:global}
Hirota, M. et~al., 2011, \enquote{Global Resilience of Tropical Forest and
  Savanna to Critical Transitions,} \emph{Science}, 334(6053): 232--235.

\bibitem[{Hunsicker et~al.(2016)}]{hunsicker:characterizing}
Hunsicker, M.~E. et~al., 2016, \enquote{Characterizing Driver-response
  Relationships in Marine Pelagic Ecosystems for Improved Ocean Management,}
  \emph{Ecological Applications}, 26(3): 651--663.

\bibitem[{Hutchings and Myers(1994)}]{hutchings:what}
Hutchings, J.~A. and R.~A. Myers, 1994, \enquote{What {C}an {B}e {L}earned from
  the {C}ollapse of a {R}enewable {R}esource? {A}tlantic {C}od, {G}adus Morhua,
  of {N}ewfoundland and {L}abrador,} \emph{Canadian Journal of Fisheries and
  Aquatic Sciences}, 51(9): 2126--2146.

\bibitem[{Jackson(2008)}]{jackson:ecological}
Jackson, J. B.~C., 2008, \enquote{Ecological Extinction and Evolution in the
  Brave New Ocean,} \emph{Proceedings of the National Academy of Sciences},
  105(Supplement 1): 11,458--11,465.

\bibitem[{Kardos et~al.(2021)}]{kardos:crucial}
Kardos, M. et~al., 2021, \enquote{The Crucial Role of Genome-wide Genetic
  Variation in Conservation,} \emph{Proceedings of the National Academy of
  Sciences}, 118(48): e2104642,118.

\bibitem[{Koopmans(1965)}]{koopmans:concept}
Koopmans, T., 1965, \enquote{On the Concept of Optimal Economic Growth,} in
  \enquote{The Economic Approach to Development Planning,} pp. 225--287,
  Chicago: Rand McNally.

\bibitem[{Kvamsdal(2022)}]{kvamsdal:optimal}
Kvamsdal, S.~F., 2022, \enquote{Optimal Management of a Renewable Resource
  Under Multiple Regimes,} \emph{Environmental and Resource Economics}, 81(3):
  481--499.

\bibitem[{Levhari and Mirman(1980)}]{levhari:great}
Levhari, D. and L.~J. Mirman, 1980, \enquote{The Great Fish War: An Example
  Using a Dynamic Cournot-Nash Solution,} \emph{Bell Journal of Economics}, 11:
  322–334.

\bibitem[{Lindig-Cisneros et~al.(2003)Lindig-Cisneros, Boyer and
  Zedler}]{lindig-cisneros:wetland}
Lindig-Cisneros, R., J.~D. K.~E. Boyer and J.~B. Zedler, 2003, \enquote{Wetland
  Restoration Thresholds: Can a Degradation Transition be Reversed with
  Increased Effort?} \emph{Ecological Applications}, 13(1): 193--205.

\bibitem[{Malhado et~al.(2010)Malhado, Pires and Costa}]{malhado:cerrado}
Malhado, A. C.~M., G.~F. Pires and M.~H. Costa, 2010, \enquote{Cerrado
  Conservation is Essential to Protect the Amazon Rainforest,} \emph{AMBIO},
  39(8): 580--584.

\bibitem[{NARP(2015)}]{narp:narp}
NARP, 2015, \enquote{NARP Maps Out \$209 Billion in Pent-Up Rail Investment,}
  NARP member newsletter.

\bibitem[{Nkuiya and Diekert(2023)}]{nkuiya:stochastic}
Nkuiya, B. and F.~Diekert, 2023, \enquote{Stochastic Growth and Regime Shift
  Risk in Renewable Resource Management,} \emph{Ecological Economics}, 208:
  107,793.

\bibitem[{Nobre and Borma(2009)}]{nobre:tipping}
Nobre, C.~A. and L.~D.~S. Borma, 2009, \enquote{`Tipping points' for the Amazon
  forest,} \emph{Current Opinion in Environmental Sustainability}, 1(1):
  28--36.

\bibitem[{Polasky et~al.(2011)Polasky, de~Zeeuw and Wagener}]{polasky:optimal}
Polasky, S., A.~de~Zeeuw and F.~Wagener, 2011, \enquote{Optimal Management with
  Potential Regime Shifts,} \emph{Journal of Environmental Economics and
  Management}, 62(2): 229--240.

\bibitem[{Ramsey(1928)}]{ramsey:mathematical}
Ramsey, F.~P., 1928, \enquote{A Mathematical Theory of Saving,} \emph{Economic
  Journal}, 38(152): 543--559.

\bibitem[{Reed(1988)}]{reed:optimal}
Reed, W.~J., 1988, \enquote{Optimal Harvesting of a Fishery Subject to Random
  Catastrophic Collapse,} \emph{Mathematical Medicine and Biology}, 5(3):
  215--235.

\bibitem[{Ren and Polasky(2014)}]{ren:optimal}
Ren, B. and S.~Polasky, 2014, \enquote{The optimal management of renewable
  resources under the risk of potential regime shift,} \emph{Journal of
  Economic Dynamics and Control}, 40: 195--212.

\bibitem[{Rose(2025)}]{rose:northern2}
Rose, G.~A., 2025, \enquote{Northern Cod Comeback: 10 Years After,}
  \emph{Canadian Journal of Fisheries and Aquatic Sciences}, 83: 1--14.

\bibitem[{Rose and Rowe(2015)}]{rose:northern}
Rose, G.~A. and S.~Rowe, 2015, \enquote{Northern Cod Comeback,} \emph{Canadian
  Journal of Fisheries and Aquatic Sciences}, 72(12): 1789--1798.

\bibitem[{Scheffer et~al.(2001)}]{scheffer:catastrophic}
Scheffer, M. et~al., 2001, \enquote{Catastrophic Shifts in Ecosystems,}
  \emph{Nature}, 413(6856): 591--596.

\bibitem[{Scott(1955)}]{scott:fishery}
Scott, A., 1955, \enquote{The Fishery: The Objectives of Sole Ownership,}
  \emph{Journal of Political Economy}, 63(2): 116--124.

\bibitem[{Selkoe et~al.(2015)}]{selkoe:principles}
Selkoe, K.~A. et~al., 2015, \enquote{Principles for Managing Marine Ecosystems
  Prone to Tipping Points,} \emph{Ecosystem Health and Sustainability}, 1(5):
  1--18.

\bibitem[{Skiba(1978)}]{skiba:optimal}
Skiba, A.~K., 1978, \enquote{Optimal Growth with a Convex-Concave Production
  Function,} \emph{Econometrica}, 46(3): 527--539.

\bibitem[{Smith(1968)}]{smith:economics}
Smith, V.~L., 1968, \enquote{Economics of Production from Natural Resources,}
  \emph{American Economic Review}, 58(1): 409--431.

\bibitem[{Storlazzi et~al.(2009)}]{storlazzi:sedimentation}
Storlazzi, C. et~al., 2009, \enquote{Sedimentation Processes in a Coral Reef
  Embayment: Hanalei Bay, Kauai,} \emph{Marine Geology}, 264: 140--151.

\bibitem[{Walters and Maguire(1996)}]{walters:lessons}
Walters, C. and J.-J. Maguire, 1996, \enquote{Lessons for Stock Assessment from
  the Northern Cod Collapse,} \emph{Reviews in Fish Biology and Fisheries}, 6:
  125--137.

\bibitem[{Worm et~al.(2006)}]{worm:impacts}
Worm, B. et~al., 2006, \enquote{Impacts of Biodiversity Loss on Ocean Ecosystem
  Services,} \emph{Science}, 314(5800): 787--790.

\end{thebibliography}
	\bibliographystyle{mybibst}
\end{document}